\documentclass[fontsize=10pt]{article}
\usepackage{amsfonts, setspace}
\usepackage{amssymb}
\usepackage{amsmath}
\usepackage{amsthm}
\usepackage{bbm}
\usepackage{algorithm, algpseudocode}
\usepackage{graphicx, float}
\usepackage{xcolor}
\usepackage{mathtools}
\usepackage{authblk} 
\usepackage{chngcntr}
\counterwithin{figure}{section}
\usepackage{caption}
\usepackage{subcaption}
\usepackage[bottom=3cm, top=3cm, left=2.5cm, right=3cm, columnsep=1cm]{geometry}
\usepackage[backend=biber, style=alphabetic, sorting=nyt, maxbibnames=99, minbibnames=99]{biblatex}
\AtEveryBibitem{\clearfield{doi}}

\usepackage[italian, english]{babel}
\usepackage[T1]{fontenc} 

\usepackage[bookmarks=true]{hyperref}
\usepackage{bookmark}

\graphicspath{{figures/}}

\usepackage[toc, page]{appendix}

\DeclarePairedDelimiter{\norma}{\lVert}{\rVert}

\newtheorem{theorem}{Theorem}[section]
\newtheorem{proposition}{Proposition}[section]
\newtheorem{remark}{Remark}[section]

\DeclareMathOperator{\dd}{d}

\DeclareRobustCommand{\varlambda}{\text{\usefont{OML}{txmi}{m}{it}\symbol{"15}}}

\def \R {{\mathbb {R}}}
\def\S{{\mathbb{S}}}
\def \N {{\mathbb N}}

\def \PO {{\mathbb {R}^3\rtimes\mathbb{S}^2}}

\title{Individuation of $3D$ perceptual units from neurogeometry of binocular cells }

\author[4,*]{M. V. Bolelli}
\author[1,2]{G. Citti}
\author[1,2]{A. Sarti}
\author[3]{S. W. Zucker}
\affil[1]{Department of Mathematics, University of Bologna, Italy.}
\affil[2]{CAMS, CNRS - EHESS, Paris, France.}
\affil[3]{Departments of Computer Science and Biomedical Engineering, Yale University, New Haven, CT, United States.}
\affil[4]{L2S, Universit\'e Paris-Saclay, CentraleSup\'elec, Gif-sur-Yvette, France.}
\affil[*]{Corresponding author: maria-virginia.bolelli@centralesupelec.fr}
\date{}

\begin{document}
\maketitle

\begin{abstract}
We model the functional architecture of the early stages of three-dimensional vision  by extending the neurogeometric sub-Riemannian model for stereo-vision 
introduced in \cite{BCSZ23}. A new framework for correspondence is introduced that integrates a neural-based algorithm to achieve stereo correspondence locally while, simultaneously, organizing the corresponding points into global perceptual units. The result is an effective scene segmentation. We achieve this using harmonic analysis on the sub-Riemannian structure and show, in a comparison against Riemannian distance, that the sub-Riemannian metric is central to the solution. 


\bigskip
\end{abstract}

\tableofcontents

\section{Introduction}

Individuating 3D perceptual units from stereo vision is a significant challenge in computer vision and neuroscience. Stereo vision systems recover 3D images from  (perspective) projections onto the two (aligned) retinae,  with each retina capturing a slightly different viewpoint of the same scene. The problem of determining the correct pairing of positions between the two retinas, known as the \textit{stereo correspondence problem}, is fundamentally underdetermined: a given pair of retinal images can correspond to multiple possible distal stimuli (\cite{P01, B21}).

In the brain, binocular visual signals are processed along the visual pathway, starting with the primary visual cortex V1. Anzai, Ohzawa and Freeman in \cite{AOF99S} showed that cells sensitive to binocular disparities in V1 perform a non-linear integration of inputs from monocular cells that are selective for orientation (interpreted in terms of tangent's orientation to the visual stimulus). Many models of monocular cells are present in literature. Let us recall the models of Hoffmann, expressed in contact geometry \cite{H89}, the 
neuromathematical model of the functional architecture of the visual cortex of  Petitot and Tondut \cite{PT99} and the model Citti and Sarti in \cite{CS06} able to describe the selectivity of cells to position and orientation in a Lie group with a sub-Riemannian metric. Moreover, neural-based stereo differential models were proposed by Zucker and collaborators in \cite{LZ06, AZ00, LZ03}, encoding both disparities in position and differences in orientation between the left and right eyes, using Frenet frames. 

More recently, a neurogeometric model for stereo vision based on sub-Riemannian geometry has been introduced \cite{BCSZ23}. The idea behind this model is to generalize the sub-Riemannian monocular model of \cite{CS06} to the functional architecture of binocular cells. Horizontal connections are described within a sub-Riemannian structure in the 3D space of position-orientation $\PO$, identified by three vector fields $Y_3, Y_4, Y_5$, where $Y_3$ encodes the tangent along 3D contours and $Y_4, Y_5$ encode changes in orientation. 
Integral curves of these vector fields model three-dimensional association fields, a formalization of the law of good continuation (\cite{W93}) and co-circularity (\cite{PZ89}) in 3D,  consistently with experimental psychophysical results (\cite{KGS05, KHK16, DW15}).

In the two-dimensional (bidimensional) case, the stochastic counterpart of association fields has been proposed as a model of cortical architecture in  \cite{SC15, BCCS14, SCS10}, and continued in \cite{FCS17}. 
The 2D analogue of the following stochastic differential equation has been considered:
\begin{equation}
\label{stochasticCurves}
\dd \Gamma(t) = Y_{3,\Gamma(t)}\dd t+ \varlambda(Y_{4,\Gamma(t)}, Y_{5,\Gamma(t)})\dd B(t), \hspace{0.5cm} \varlambda \in \R,
\end{equation}
where $B(t)$ is 2-dimensional Brownian motion. 
The probability density of this system, integrated in time, provides a kernel $J_\varlambda$, a fundamental solution of a suitable forward Kolmogorov equation.
This kernel can be inserted as a facilitator pattern in the evolution of a mean field equation proposed as a model of cortical activity, originally by Wilson-Cowan \cite{WC72, WC73}, by Amari \cite{A72}, and then by Ermentrout and Cowan \cite{EC79, EC80} and Bresslof and Cowan \cite{BCGTW02, BC03}. Other similar models, that account for cortical delays, were proposed by Faye and Faugeras \cite{FF10}. Finally, Sarti and Citti \cite{SC15} proposed applying a spectral analysis method to recover 2D perceptual units.
Following these works, we describe the propagation of the visual signal impulse along cortical connectivity by the following integro-differential equation in the 3D (perspective) space of positions and orientations:
\begin{equation}
\label{meanFieldIntro}
\frac{da(\xi, t)}{dt}= - a (\xi, t)+ \int \mu J_\varlambda(\xi, \xi')a(\xi', t)d\xi' + h(\xi, t)
\end{equation}
where $\xi \in \PO$ , $h(\xi,t)$ is the feedforward input, $\mu \in \R$.

To achieve the extension to 3D perceptual units, we confront the fact that stereo correspondence is a global problem: {\em a priori}, each point on the left retina could -- in principle -- be coupled with any point on the right one. The correct coupling is the one that produces a ``realistic'' 3D perceptual unit; that is, a coupling of points in each retina that obey the Gestalt laws of perceptual organization in 3D. Inspired by this, we consider all possible pairs of (discrete image) points to obtain a 3D cloud of points. To extract perceptual units from the cloud, i.e. to group the points into global solutions to the correspondence problem, we extend an observation of Perona and Freeman \cite{PF98}, and show that the coherent sets are represented by the eigenvectors of \eqref{meanFieldIntro}. The  ``salience'' of each perceptual unit is ranked by its associated eigenvalue, with the largest eigenvalues signaling the most salient (parts of) objects in the scene. The 3D law of good continuation, encoded in the kernel, is thus applied to the space of possibilities and is leveraged to solve the global correspondence problem.

\paragraph*{Structure of the paper.}
The paper is organized as follows. 
Section 2 covers preliminary results: we recall the stereo problem, the neurogeometric model proposed in \cite{BCSZ23} for stereo vision in the space of 3D positions-orientations, and the law of good continuation in 3D. 
Section 3 is devoted to discussing a stochastic model for cortical connectivity. This model is used in a mean-field equation to express the evolution of neural activity. We perform a stability analysis using a variation of the Lyapunov method.
In Section 4 we consider spectral clustering and dimensionality reduction results for individuating 3D percepts, performing segmentation and proposing a solution for the stereo correspondence. In Section 5, we present the numerical implementation and the results. First, we numerically compute the connectivity kernel and then, we use it to implement spectral analysis to individuate 3D perceptual units and solve the correspondence problem. We end this section with a comparison with a ``Gaussian-type'' kernel to emphasize the role of our chosen metric. Section 6 is devoted to the conclusions.

\section{Background}

\subsection{The stereo problem: state of the art and improvement with neurogeometry}
\label{sec:stereo_problem}

The stereo problem involves reconstructing the three-dimensional visual scene from its perspective projection through the left and right optical centers.
Typically, the primary objective is to accurately match corresponding points $ Q_L = (x_L, y)$ and $ Q_R = (x_R, y) $ on the parallel retinal planes $( y = y_L = y_R )$. Traditional stereo algorithms rely on intensity-based matching 
( e.g. \cite{birchfield1998pixel, ng2019solving, karimi2022stereo} )
 and feature-based matching 
(e.g. \cite{MP79, alwan1996automatic, zhu2017hybrid, LZ03}) techniques. To facilitate the process of finding the correct matches, several constraints are introduced, mainly exploiting geometric-based approaches. Extensive textbooks on the subject can be found in \cite{F93, FL01, HZ03, SR11}. 
Once these points are correctly paired, they can be projected back into the environmental space to obtain $Q = (r_1, r_2, r_3) \in \mathbb{R}^3 $.

However, in this paper, we adopt a different approach. Instead of first pairing points and then reconstructing the 3D scene, we reverse the process. Initially, we consider all possible corresponding pairs of points between the left and right images (namely points sharing the same y-coordinates since the retinal planes are parallel) and their 3D representations, obtaining a cloud of points. We extract the 3D perceptual unit from this cloud, which represents the desired 3D object, leveraging the 3D good continuation law and its underlying cortically inspired sub-Riemannian geometry (\cite{BCSZ23b, BCSZ23}). Identifying these 3D percepts will solve the stereo correspondence problem, automatically determining the correct pairing on the retinal planes.

\subsection{Neurogeometry of binocular vision and good continuation}
\label{sec:sub_riem}
The behavior of monocular simple cells selective for position and orientation in V1 and their functional architecture have been modeled using tools from sub-Riemannian geometry in the works of Petitot and Tondut \cite{PT99} and Citti and Sarti \cite{CS06}. Meanwhile, Zucker et al. \cite{AZ00, LZ06, LZ03} proposed a mathematical model for stereo vision based on neural mechanisms of selectivity to position, orientation, and curvature of the visual stimulus, expressed through Frenet differential geometry.

From a neural point of view, it seems that neurons driven by binocular input in the primary visual cortex integrate input from monocular (orientation-selective) cells (\cite{AOF99, STRTF22}). Precisely, they integrate positions on the left and right retinas and the corresponding orientations $\theta_L$ and $\theta_R$ \footnote{We recall that the angle $\theta_i$ denotes the orientation of the visual signal's boundary.}, to obtain perception in the 3D external space.

Inspired by these neural findings and by the previously mentioned models \cite{AZ00, LZ06, LZ03, CS06}, 
a neurogeometric model of binocular  vision 
has been proposed in 
 \cite{BCSZ23}. 
The visual signal projected onto the two retinas, together with the orientation of its boundaries, is lifted into a 5D space, coding position and orientation in the 3D external space:  $\mathbb{R}^3_{(r_1, r_2, r_3)}\rtimes \mathbb{S}^2_{(\theta, \varphi)}.$
Then, the lifted signal is propagated along integral curves of the following vector fields:
\begin{equation}
\label{eq:orthonormal_frame}
\begin{aligned}
Y_{3, \xi} &= \cos\theta\sin\varphi \partial_1+ \sin\theta\sin\varphi\partial_2+\cos\varphi \partial_3, \
Y_{4, \xi }&= -\frac{1}{\sin\varphi}  \partial_\theta, \
Y_{5,\xi} &= \partial_\varphi,
\end{aligned}
\end{equation}
locally using the chart $\theta \in (0, 2\pi), \varphi \in (0, \pi)$, and $\xi \in \mathbb{R}^3_{(r_1, r_2, r_3)}\rtimes \mathbb{S}^2_{(\theta, \varphi)}$. 
The vector field 
$Y_{3}$ encodes the tangent of the stimulus, 
the vector field $Y_5$ involves orientation in the depth direction, while $Y_4$ involves orientation on the fronto-parallel plane.

Denoting with $\vec{Y}_{i, \xi}$ the direction of the vector field $Y_{i, \xi}$ at a point $\xi \in \PO$, integral curves with constant coefficients are defined by the differential equation:
\begin{equation}
\label{eq:integralCrv}
\dot \Gamma(t) = \vec{Y}_{3,\Gamma(t)}+ c_1\vec{Y}_{4, \Gamma(t)}+ c_2 \vec{Y}_{5, \Gamma(t)}, \quad c_1, c_2 \in \mathbb{R},
\end{equation}
an example of which is displayed in Figure \ref{fig:association_field} (b).

Integral curves of the vector field $Y_3$ express continuation in the  direction of the boundary, while $Y_4$ and $Y_5$ express deviation from this line. For this reason they are strictly related to the Gestalt law of good continuation. This principle expresses  the perceptual organization of oriented elements in a visual scene (see for example \cite{WEKP12}). In 2D these relationships have been formalized with the concept of co-circularity by Parent and Zucker in \cite{PZ89, BZ04}, and the concept of association fields  by Field, Hayes, and Hess in \cite{FHH93}. These concepts are well modeled by integral curves of the sub-Riemannian monocular structure proposed by Citti and Sarti in \cite{CS06} (see Figure \ref{fig:association_field}, (a)).

In 3D, the geometric relationships among orientations within a three-dimensional setting have been formalized with the theory of 3D relatability (\cite{KGSYM05}). Curves linking these oriented points exhibit smoothness, monotonic behavior, and regularity (\cite{HHK97, DW15}). Additionally, the strength of relatable edges in coplanar planes with the initial edge must meet the relations of the bi-dimensional association fields \cite{KGSYM05}.
In  \cite{BCSZ23} 
the family of integral curves \eqref{eq:integralCrv} has been proposed as a sufficent model for associations in 3D (see Figure \ref{fig:association_field} (b) and (c)).

 \begin{figure}[tbh]
\centering
\begin{subfigure}[b]{0.18\textwidth}
         \centering
    \includegraphics[width=\textwidth]{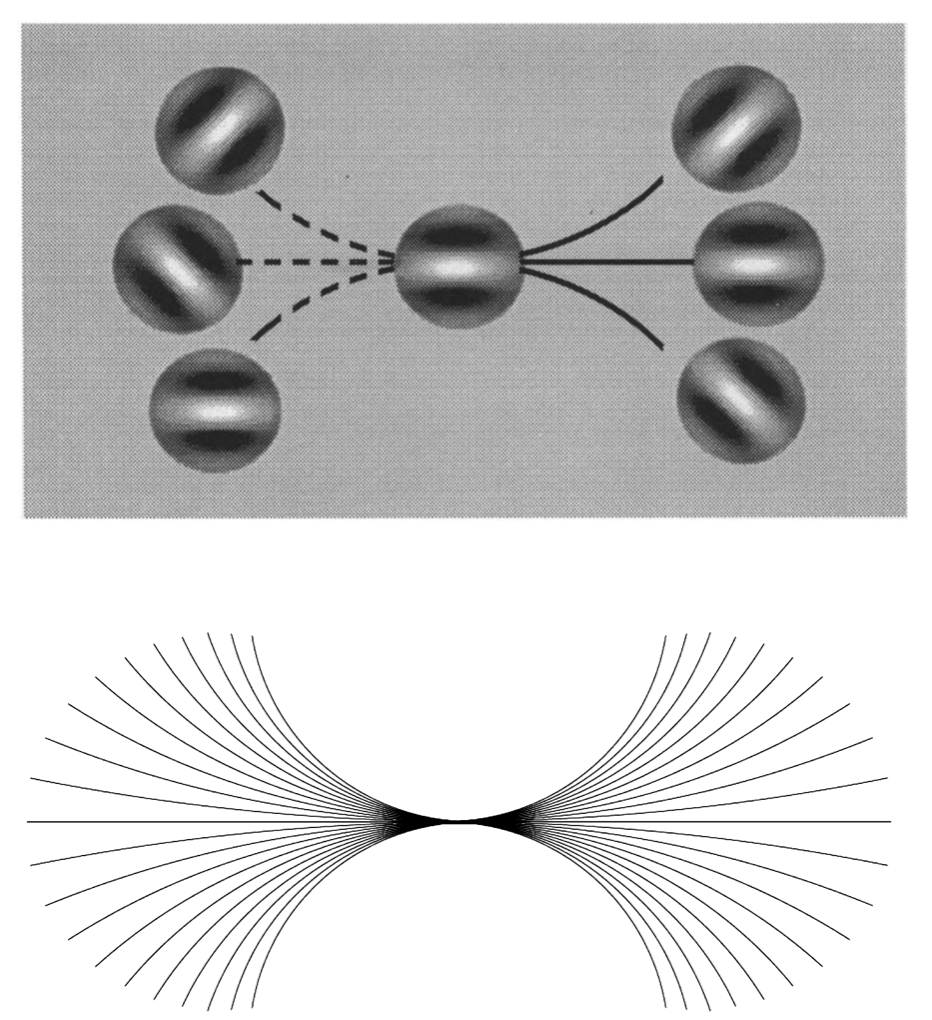} 
         \caption{}
\end{subfigure}
\begin{subfigure}[b]{0.48\textwidth}
         \centering
         \includegraphics[width = \textwidth]{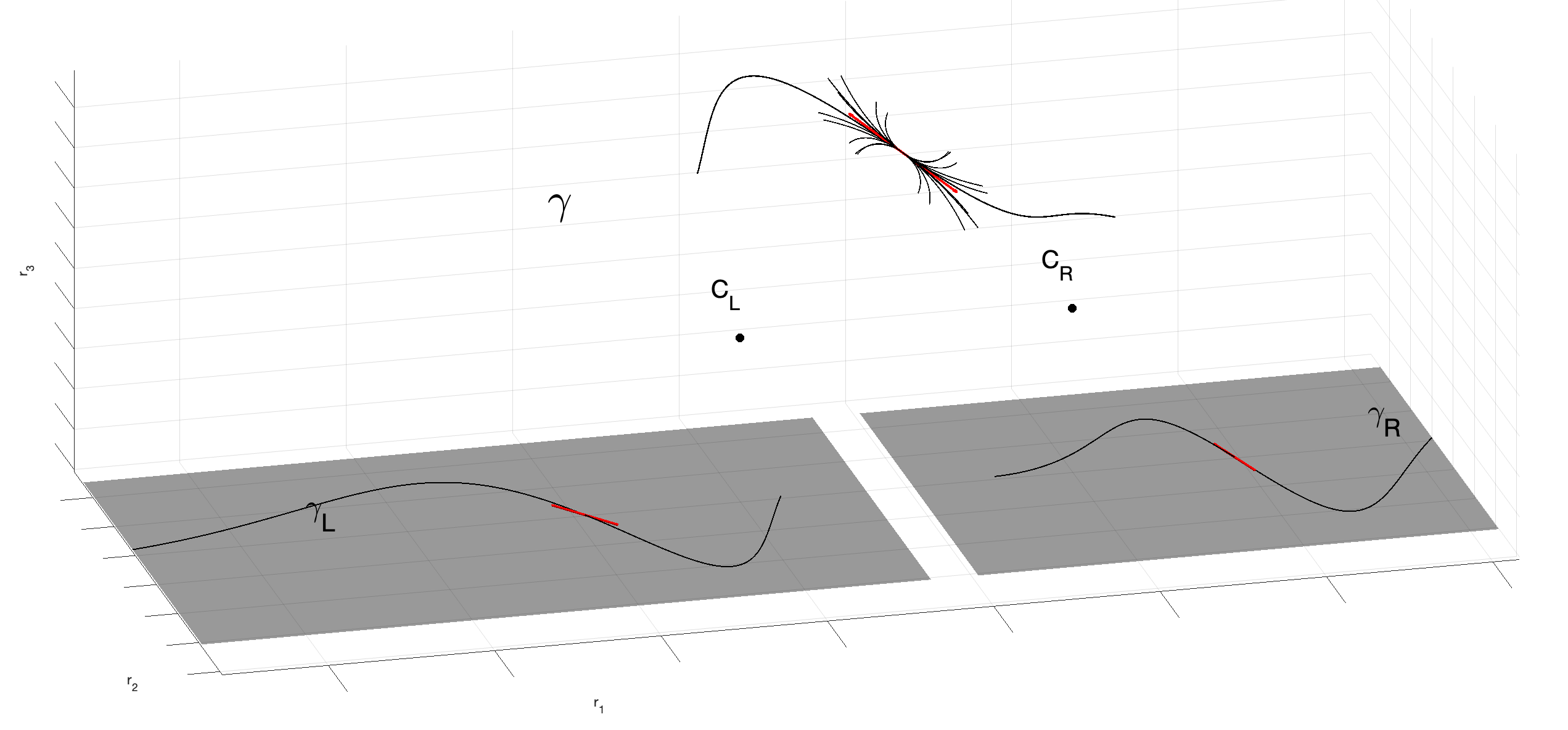} 
         \caption{}
\end{subfigure}
\begin{subfigure}[b]{0.3\textwidth}
         \centering
         \includegraphics[width=\textwidth]{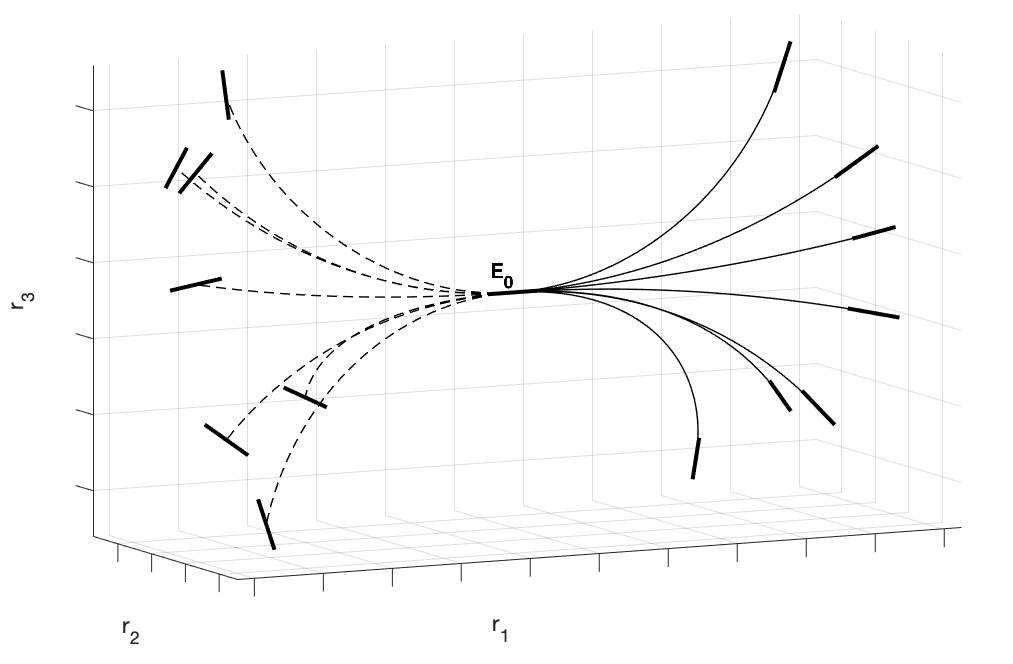}
         \caption{}
\end{subfigure}
      \caption{ Display of connectivity. (a) Field Hayes and Hess association field (top) and 2D integral curves of the Citti-Sarti model \cite{CS06} (bottom). (b) General fan of integral curves described by equation \eqref{eq:integralCrv} with varying $c_1$ and $c_2$ in $\R$, enveloping a curve $\gamma \in \R^3$. (c) Fan of 3D relatable points connected by integral curves \eqref{eq:integralCrv}.}
      \label{fig:association_field}
\end{figure}

\section{Connectivity kernel and mean field equation}
\label{sec:kernel-mean_field}
We investigate the transition from association fields, which represent a psychophysical abstraction, to measurements that quantify the interaction strength among the position-orientation elements that form 3D perceptual units, specifically three-dimensional contours. To achieve this, we extend a technique previously introduced in \cite{SCS10, BCCS14, SC15} to the position-orientation space $\PO$.

\subsection{Stochastic paths and Fokker-Planck equation}
\label{ssec:kernel}
The probabilistic version of the association fields is used in the two-dimensional case as a model of connectivity for cortical cells, for example in the works \cite{AugZ00, PT99, SC15, AFCSR17, DF10, BCCS14}, based on the pioneering work of Mumford \cite{M94}. We apply this approach to the three-dimensional case, probabilistically interpreting the 3D extension of the law of good continuation: entities described by similar local orientations are more likely to belong to the same perceptual unit. In particular, we consider an equation predominated by transport in the $Y_3$ direction and diffusion in the orientation variables, similar to what has been done in \cite{MS09, MS13}.

The result is an It\^o stochastic process $\Gamma(t)$ satisfying the stochastic differential Langevin equation:  
\begin{equation}
\label{SDE}
\dd \Gamma(t) = \vec{Y}_{3,\Gamma(t)}\dd t + \varlambda (\vec{Y}_{4,\Gamma(t)}, \vec{Y}_{5,\Gamma(t)})\dd B(t), 
\end{equation}
with $B(t) = (B_1(t), B_2(t))^T$ a two-dimensional Brownian motion and $\varlambda \in \R$ scalar parameter. By abuse of notation, $Y_3$ is interpreted as a deterministic coefficient and it is called the \textit{drift} coefficient, while we call $\sigma(\Gamma(t))\coloneqq(\vec{Y}_{4,\Gamma(t)}, \vec{Y}_{5,\Gamma(t)})$ the \textit{diffusion} coefficient of the stochastic differential equation (SDE) \eqref{SDE}. 

\begin{remark}
Locally in the chart of section \ref{sec:sub_riem}, these coefficients are smooth functions with locally bounded derivatives. Therefore, classical results on uniqueness and existence of a solution to the SDE \eqref{SDE} hold, since the drift and diffusion coefficient terms are Lipschitz continuous.
\end{remark}
 
It\^o's formula allows us to consider the density law $\rho_\varlambda$ of the stochastic process $\Gamma$, which is a distributional solution to the Fokker-Planck (or forward Kolmogorov) equation $\partial_t + \mathcal{L} =\delta $ (e.g, see \cite{O03}), where:
\begin{equation}
\label{FP}
\mathcal{L}f(t, \xi) = -\sum_{i =1}^3 (\vec{Y}_{3,\xi})_i\partial_{i} f (t, \xi)+ \frac{\varlambda}{2}\sum_{i, j \in \{\theta, \varphi\}} (\sigma(\xi)\sigma(\xi)^T)_{i,j}\partial_i\partial_j f(t, \xi);
\end{equation}
the  Kolmogorov operator $\mathcal{L}$ can be written in terms of vector fields \eqref{eq:orthonormal_frame}: 
\begin{equation}
\label{fKoperator}
\mathcal{L} = -Y_3 + \varlambda(Y_4^2 + Y_5^2).
\end{equation}
Since the vector fields involved satisfy the H\"ormander's rank condition (\cite{DF11}), the operator $\mathcal{L}$ is hypoelliptic. So, there exists a fundamental solution associated with the forward Kolmogorov equation, with the property of being smooth out of the pole (see H\"ormander theorem \cite{H67}).

\subsubsection{Time independent kernel}
\label{sec:independent_time_kernel}

The fundamental solution for the operator $\partial_t + \mathcal{L}$, with $\mathcal{L}$ defined in \eqref{fKoperator}, corresponds to the probability density function $ \rho_\varlambda(\xi,t; \xi_0, t_0)$ associated with the stochastic process satisfying equation \eqref{SDE}. More precisely, it expresses the probability of having reached the point $\xi$ after having evolved the equation \eqref{SDE} up to time $t$, starting from $\xi_0$ at time $t_0$. 

The main goal of this section is to identify a time-independent probability, to characterize each point of the space in terms of the paths \eqref{SDE} that reach it independently of the value of the evolution temporal parameter. 
It is therefore worth noting that these density functions $\rho_\varlambda$ locally admit exponential-type estimates, which decay to zero for all ``evolutional'' times.

\begin{remark}
\label{rmk:gauss_est}
The operator $\mathcal{L}$ of \eqref{fKoperator} belongs to a class of hypoelliptic operators that locally admit estimates for their fundamental solution. Rotschild and Stein in \cite{RS76} obtained optimal local estimates of the kernels, while Gaussian upper and lower bounds have been proven in  \cite{BLU02, BP07, CP08}. 
\end{remark}

In the literature, it is common to integrate over time to obtain such a time-independent fundamental solution as shown in \cite{BLU02} or in \cite{Cocci14, FCS17}.

\begin{proposition}
Let  $\xi$ be a point in $ \PO$ and $\rho_\varlambda$ the fundamental solution of the operator $\partial_t + \mathcal{L}$ with $\mathcal{L}$ defined in \eqref{fKoperator}, and pole in $(\xi_0, 0) $, with $t_0 = 0$ without loss of generality. Then
\begin{equation}
\label{eq:time_indep_kernel}
J_\varlambda (\xi, \xi_0) = \int_{\R_+} \rho_\varlambda( \xi, t; \xi_0, 0) dt
\end{equation} 
is fundamental solution for the operator $\mathcal{L}$:
\begin{equation}
\mathcal{L}J_\varlambda(\xi, \xi_0) = \delta_{\xi-\xi_0}.
\end{equation}
\end{proposition}

\proof
Let $u$ be a smooth function on $\PO$ with compact support, independent from the time variable $t$. Then: 
\begin{equation}
\begin{aligned}
-u(\xi_0) = &  \int_{\R_+\times\PO} \rho_\varlambda( \xi, t; \xi_0, 0) (\partial_t + \mathcal{L}) u(\xi) d \xi d t\\
 = &  \int_{\R_+\times\PO} \rho_\varlambda( \xi, t; \xi_0, 0) \mathcal{L} u(\xi) d \xi d t\\
 = & \int_{\PO}\left (\int_{\R_+}\rho_\varlambda( \xi, t; \xi_0, 0)d t \right) \mathcal{L} u(\xi) d \xi\\
 = & \int_{\PO} J_\varlambda( \xi, \xi_0)\mathcal{L} u(\xi) d \xi\\
\end{aligned}
\end{equation}
with $J_\varlambda( \xi, \xi_0)= \int_{\R_+} \rho_\varlambda( \xi, t; \xi_0, 0)d t$.
\endproof

Alternatively, another way to deal with the time parameter is to make a further assumption on the random walk, by supposing it has an exponentially distributed traveling time. Then the time-independent fundamental solution is obtained by taking the Laplace transform of the probability density, as done for example in \cite{DF11, PD17}. 

\begin{remark}
The partial differential equations mentioned have been extensively studied in various works (\cite{DBM19, PD17, PSMD15, DF11, DCGD11, D16}), focusing on the analysis of HARDI (High Angular Resolution Diffusion Imaging) images in $\mathbb{R}^3$. These images demonstrate roto-translation invariance in three dimensions, indicating that they remain invariant under rigid-body motions within the $SE(3)$ group.
The vector fields \eqref{eq:orthonormal_frame} are not left-invariant, leading to a model that is not entirely isotropic. This aligns with the fact that in three-dimensional vision, there exists a direction that is not perfectly consistent with this invariance (\cite{KHK16}).\end{remark}

\subsection{Mean field equation and spectral analysis}

The evolution of the neuronal population activity has been modeled through a mean field equation by many authors, for example \cite{ EC80, BC03, FF10, SC15}; it was first proposed by Amari \cite{A72} and Wilson and Cowan \cite{WC72}. According to these papers, the result of the propagation can be described (without considering the delays in the transmission of the signal) by the following integro-differential equation in the 3D perceptive space of positions and orientations: 
\begin{equation}
\label{eq:meanField}
\frac{da(\xi, t)}{dt}= -\alpha a (\xi, t)+ \varrho\left(\int_{\PO}\mu J_\varlambda(\xi, \xi')a(\xi', t)d\xi' + h(\xi, t)\right),
\end{equation}
where $\xi$ is a point of $\PO$, $t > 0$, $d\xi = d x 
d\sigma(n) = dx\sin\varphi d\theta d \varphi
$, with $dx$ Lebesgue measure in $\R^3$ and $d\sigma$ spherical measure on $\S^2$, $\alpha \in \R$ represents the decay of activity, and $h(\xi,t)$ is the feedforward input.

The function $\varrho$ is the firing rate function of the population and it has a piecewise linear behavior, as proposed in \cite{SC15}, building upon \cite{KB10}: 
\begin{equation}
\varrho(s) = 
\begin{cases}
0,  &s \in ]-\infty, c-\frac{1}{2\gamma}[\\
\gamma(s-c) + \frac{1}{2}   &s \in [c-\frac{1}{2\gamma}, c+\frac{1}{2\gamma}]\\
1,  &s \in ]c-\frac{1}{2\gamma}, + \infty[\\
\end{cases}
\end{equation}
where $\gamma$ is a real number that represents the slope of the linear regime and $c$ is the half-height threshold. 
The parameter $\mu$ is a coefficient of short-term synaptic facilitation, while the kernel $ J_\varlambda(\xi, \xi')$ is the contribution of the cortico-cortical connectivity introduced in the previous sections through equation \eqref{eq:time_indep_kernel}. Bidimensional models \cite{SCS10, BCCS14} deal with symmetric kernels, in order to consider neural reciprocal connections.

\begin{remark}[Symmetric kernel]
\label{symmetricKernel}
The kernel $J_\varlambda(\xi, \xi')$ introduced through equation \eqref{eq:time_indep_kernel} is non-symmetric because, in $\R^3$, it is squeezed in the direction orthogonal to $Y_3$ (\cite{NSW85}). To address this, we define its symmetrization as follows:
\begin{equation}
J_S(\xi, \xi') = \frac{J_\varlambda(\xi, \xi')+J_\varlambda(\xi', \xi)}{2}.
\end{equation}
This symmetrization involves averaging the fundamental solution associated with \eqref{SDE} with the fundamental solution of the same operator, but with an angular shift in the opposite three-dimensional direction. 
 This rotation changes the drift term $Y_3$ into $-Y_3$, thereby transforming the forward Kolmogorov equation into the corresponding backward equation.
\end{remark}
 
Finally, it is important to address the integration set for equation \eqref{eq:meanField}. By leveraging the findings from \cite{SC15}, we can partition the space $\PO$ into a subdomain $\Omega \subset \PO$ that is activated by the input. Specifically, we define $\Omega$ as:
\begin{equation}
\Omega = \{ \xi \in \PO \mid h(\xi) = c \}.
\end{equation}
In this context, we assume that the feedforward input function $h$ takes on only two values: $0$ and a constant $c$. Consequently, there is a complementary set where the activity is negligible. Additionally, we assume that $\Omega$ has finite measure with respect to the introduced measure $d\xi$; further details can be found in \cite{SC15}.
Then, the mean field activity equation \eqref{eq:meanField} reduces to: 
\begin{equation}
\label{eq:meanField_Restriction}
\frac{da(\xi, t)}{dt}= -\alpha a (\xi, t)+ \gamma\left(\int_{\Omega}\mu J_S(\xi, \xi')a(\xi', t)d\xi' + c \right), \text{ for } \xi \in \Omega. 
\end{equation}

\subsubsection{Existence and uniqueness of a solution}

We consider $L^2(\Omega, \R)$, the space of the square-integrable functions from $\Omega$ to $\R$ with  the usual inner product with respect to the measure $d\xi$, product of Lebesgue measure in $\R^3$ and the spherical measure on $\S^2$. 
We define $f(a) := -\alpha a+A(a) +c$, with 
\begin{equation}
\label{kernelOperator}
A(a) = \gamma\int_{\Omega}\mu J_S(\xi, \xi')a(\xi', t)d\xi',
\end{equation}
and we denote with $I$ a closed interval on the real line containing $0$. We consider a mapping $a : I \longrightarrow L^2(\Omega, \R)$, so equation \eqref{eq:meanField_Restriction} can be recast as a Cauchy Problem:
\begin{equation}
\label{CP}
\begin{cases}
a'(t) =f(a)\\
a(t_0)= a_0\\
\end{cases} 
\end{equation} with initial value $a_0 \in L^2(\Omega, \R)$.

Classical instruments of functional analysis (see for example \cite{HL13}) can then be applied to study the existence and uniqueness of the solution to the Cauchy Problem \eqref{CP} in Hilbert spaces. We just need to show that the function $f$ is well defined and Lipschitz.

\begin{proposition} If $J_S(\xi, \xi') \in L^2(\Omega \times \Omega, \R)$, the function $f$ is well defined and $f(a) \in L^2(\Omega, \R)$ for all $a \in L^2(\Omega, \R)$.
\end{proposition} 

\begin{proof} 
Performing a direct computation: 
\begin{equation}
\begin{aligned}
&\norma{f(a)}_{L^2(\Omega, \R)}^2\leq  \\
&\norma{a}^2 + \norma{A(a)}^2 + \norma{c}^2 + 2 \norma{a}^2\norma{A(a)}^2 + 2 \norma{c}^2\norma{A(a)}^2 + 2 \norma{c}^2\norma{a}^2,
\end{aligned}
\end{equation}
exploiting the Cauchy-Schwarz inequality. Since $a$ and $c \in L^2(\Omega, \R)$, we only need to prove that $\norma{A(a)}_{L^2(\Omega, \R)} < \infty$:
\begin{equation}
\begin{aligned}
\norma{A(a)}^2_{L^2(\Omega, \R)} & =\int_{\Omega} \left( \int_\Omega J_S(\xi, \xi')a(\xi', t)d\xi'\right)^2 d\xi\\
& \leq \int_{\Omega} \left( \int_{\Omega}J_S^2(\xi, \xi')d\xi' \int_{\Omega} a^2(\xi', t) d\xi'\right)d \xi \\
&\leq\norma{J_S}_{L^2(\Omega\times\Omega, \R)}^2\norma{a}_{L^2(\Omega, \R)}^2
\end{aligned}
\end{equation}
using again the Cauchy-Schwarz inequality.
\end{proof}

\begin{remark}
\label{locallyIntegrableKernel}
The kernel $J_S(\xi, \xi')$ is obtained as fundamental solution of an hypoelliptic operator. Since local estimates in terms of Gaussian upper and lower bounds are provided for example in \cite{BoPo07}, it follows that this kernel is locally integrable, in particular locally square integrable. Furthermore, the operator \eqref{kernelOperator} is linear, bounded, and compact on the measure space of definition, see \cite{C19}.
\end{remark}

\begin{proposition}
Let $J_S \in L^2(\Omega\times\Omega, \R)$, then we have that the function $f$ is Lipschitz in the variable $a \in L^2(\Omega, \R)$. 
\end{proposition}

\begin{proof} 
We have 
\begin{equation}
\begin{aligned}
\norma{f(a_1)-f(a_2)}_{L^2(\Omega, \R)} = & \norma{-a_1+a_2 + A(a_1)-A(a_2)}_{L^2(\Omega, \R)} \\
\leq& \norma{a_1-a_2}_{L^2(\Omega, \R)}+\norma{A(a_1)-A(a_2)}_{L^2(\Omega, \R)}\\
\leq & (1+M)\norma{a_1-a_2}_{L^2(\Omega, \R)}\\
\end{aligned}
\end{equation}
where we first use the triangular inequality, and then the properties of Remark \ref{locallyIntegrableKernel}. In particular, $M$ is the Lipschitz constant for the operator $A$ defined through \eqref{kernelOperator}.
\end{proof}

\subsubsection{Stability analysis}
\label{sec:stability}

The neuronal population activity $a(\xi, t)$ represents the average excitation of cells resulting from an initial input (such as a visual stimulus) and from interactions among the cells. Perceptual units are identified by the cells most involved in this activation process, and typically, the perception of these units is independent of time.

The stationary solutions $a_0$ of equation \eqref{eq:meanField_Restriction}, which are time-independent, satisfy the following equation:
\begin{equation}
\label{meanFieldStable}
  -\alpha a_0 (\xi) + \gamma \left(\int_{\Omega} \mu J_S(\xi, \xi') a_0(\xi') \, d\xi' + c \right) = 0.
\end{equation}
We now focus on the study of stable stationary states, as these are crucial for ensuring the stability of the perception of visual units.

\paragraph{Lyapunov method in the space of position and orientation.}

The stability of the Cauchy problem \eqref{CP} associated with a stationary state $a_0$ can be analyzed using a technique introduced by Faye-Faugeras \cite{FF10}, which is based on Lyapunov functionals.

We consider a small perturbation around the stationary state $u := a - a_0$, so that the linearized equation associated with \eqref{eq:meanField_Restriction} can be written as:
\begin{equation}
\label{linMeanField}
\frac{\partial u(\xi, t)}{\partial t}= -\alpha u (\xi, t)+ \gamma\int_{\Omega}\mu J_S(\xi, \xi')u(\xi', t)d\xi',  \text{ for } \xi \in \Omega.
\end{equation}

\begin{proposition}
If $ \norma{J_S}_{L^2(\Omega^2, \R)}< \frac{\alpha}{\mu\gamma}$ then $a_0$ it is stable steady state for \eqref{CP}. 
\end{proposition}

\proof
It follows from \cite{FF10} that if $0$ is asymptotically stable for \eqref{linMeanField}, then $a_0$ is asymptotically stable for \eqref{eq:meanField_Restriction}. So, we consider the Lyapunov functional defined as the quadratic form: 
\begin{equation}
\label{Lyapunov}
V(u) = \frac{1}{2} \int_{\Omega} u^2(\xi) d\xi.
\end{equation} This function is positive definite; then we analyze its time derivative along the trajectories of the system:
\begin{equation}
\begin{aligned}
dV(u)(\tilde A(u)) &= -\int_\Omega\alpha u(\xi)u(\xi)d\xi+\int_{\Omega} u(\xi)\mu\gamma\int_{\Omega} J_S(\xi, \xi')u(\xi')d\xi' d \xi\\
& \leq -\alpha\norma{u}^2_{L^2(\Omega, \R)} + \mu\gamma\norma{J_S}_{L^2(\Omega^2, \R)}\norma{u}^2_{L^2(\Omega, \R)} \\
& = (-\alpha+\mu\gamma\norma{J_S}_{L^2(\Omega^2, \R)})\norma{u}^2_{L^2(\Omega, \R)},
\end{aligned}
\end{equation}
from which the premise follows. 
\endproof

\begin{remark}
\label{rmk:eigenv}
Since equation \eqref{linMeanField} is linear, we can further investigate the stability of the stationary solution by examining the eigenvalue problem associated with the linear operator:
\begin{equation}
\label{meanFieldEigv}
L u := -\alpha u + \gamma \mu \int_{\Omega} J_S(\xi, \xi') u(\xi') \, d\xi' = \lambda u \quad \iff \quad \int_{\Omega} J_S(\xi, \xi') u(\xi') \, d\xi' = \tilde{\lambda} u,
\end{equation}
where $\tilde{\lambda} = \frac{\lambda + \alpha}{\gamma \mu}$.
In general, the system is stable if $\lambda < 0$. Thus, for stability, we require $\tilde{\lambda} < \frac{\alpha}{\gamma \mu}$. 
\end{remark}

\section{Perceptual units individuation and the correspondence problem}
\label{sec:percepts-and-stereo}
In this section, we present our approach to solve the correspondence problem. Inspired by cortical mechanisms, our goal is to identify perceptual units in the 5D space of position and orientation.
More specifically, following the methodology outlined in \cite{SC15}, we investigate how the spectral analysis performed by the neuronal population studied in the previous section can be applied to identify 3D perceptual units. This results in a grouping operation that is closely related to principles of spectral clustering and dimensionality reduction. In this framework, the salient objects in the scene correspond to the eigenvectors associated with the largest eigenvalues.
We then use these grouping techniques to identify the perceptual units within the context of stereo vision. The result completes the matching problem together with a global organization of the matches.

\subsection{Spectral analysis and emergence of perceptual units}

In applications, dealing with a finite number of points is typical, resulting in a discrete structure for the visual stimulus. We consider a set $\Omega$ of $k$ points, each defined by position-orientation coordinates:
\begin{equation}
    \Omega \coloneqq \{\xi_i = (r_{1i}, r_{2i}, r_{3i}, \theta_{i}, \varphi_{i}) \in \PO, \text{ for } i = 1, \ldots k\}.
\end{equation}
These points correspond to the binocular cells activated by the stimulus. In this context, the discrete mean field equation \eqref{linMeanField} becomes:
\begin{equation}
\label{meanFieldD}
\frac{du(\xi_i, t)}{dt}= -\alpha u (\xi_i, t)+ \gamma\mu\sum_{j=1}^k J_S(\xi_i, \xi_j)u(\xi_j, t).
\end{equation}
Here,  $J_S$ is  the kernel function which is reduced to a symmetric matrix $\textbf{J}$ of dimension $k \times k$, with entries given by:
\begin{equation}
\textbf{J}_{ij}= \gamma\mu J_S(\xi_i, \xi_j).
\end{equation}
The corresponding eigenvalue problem \eqref{meanFieldEigv} thus simplifies to:
\begin{equation}
\textbf{J}u= \tilde \lambda u. 
\end{equation}
This matrix $J$ is analogous to the \textit{affinity matrix} used in spectral clustering and dimensionality reduction problems, as introduced in numerous studies (e.g., \cite{RS00, BN03, CL2006, PF98, MS01, SM00}). It quantifies how similar or ``affine'' the points are to each other.

\subsubsection{Spectral clustering, dimensionality reduction and salient objects in the visual scene}

It is possible to consider the set of points $\Omega$ as vertices of a weighted graph, where the
 weights of the edges connecting these points are given by the affinity matrix $\textbf{J}$. Since this affinity matrix is symmetric, it has been shown in \cite{PF98} that the leading eigenvectors can be used as an indicator vectors for preliminary grouping tasks. 
 
\begin{remark}
Perona and Freeman (\cite{PF98}) introduced a method to characterize a visual scene by approximating the matrix $\mathbf{J}$ as a sum of rank-1 matrices. Each rank-1 matrix is formed by the outer product of a vector $\mathbf{p}$. The first rank-1 approximation is determined by minimizing the Frobenius norm:
\begin{equation}
\mathbf{p}_1 = \text{argmin}_{\mathbf{p}} \sum_{i, j = 1}^k (\mathbf{J}_{ij} - p_i p_j)^2.
\end{equation}
It was shown that $\mathbf{p}_1$ corresponds to the eigenvector $\mathbf{v}_1$ of $\mathbf{J}$ scaled by the square root of the largest eigenvalue $\lambda_1$: $\mathbf{p}_1 = \lambda_1^{1/2} \mathbf{v}_1$. This approach is iteratively applied to identify all significant eigenvectors, typically resulting in a reduction of the problem's dimensionality to fewer than $k$ eigenvectors.
\end{remark}

As a result, the problem of grouping is reduced to the spectral analysis of the affinity matrix $\mathbf{J}$, where the salient objects in the scene correspond to the eigenvectors with the largest eigenvalues. In section \ref{sec:kernel-mean_field}, we have shown that  spectral analysis can be implemented by the neuronal population binocular cells in the primary visual cortex. So, this eigenvector provides an initial indication of how to separate the data into groups. However, this basic method can be affected by noise and nonlinear distributions, leading to clustering errors (\cite{W99}).

\subsubsection{Random walk normalization}

The approach was further refined, particularly in relation to problems such as minimal graph cuts (\cite{SM00}), by performing spectral analysis on a suitably normalized affinity matrix.

Using the normalization proposed in \cite{MS01}, the affinity matrix $\textbf{J}$ is transformed into the transition matrix $\textbf{P}$ of a Markov process, via row-wise normalization. Specifically, if $\textbf{D}$ is the diagonal degree matrix with elements
\begin{equation}
\textbf{D}_{ii}= \sum_{j= 1}^k \textbf{J}_{ij}, \text{ for } i = 1 \ldots k , 
\end{equation}
then, the normalized affinity matrix $\textbf{P}$ is given by 
\begin{equation}
\label{eq:norm}
\textbf{P} = \textbf{D}^{-1}\textbf{J}.
\end{equation}

\begin{remark}
The matrix $\textbf{P}$ is generally not symmetric. However, its eigenvalues $\{\lambda_i\}_{i = 1}^k$ are real and lie between $0$ and 1. The eigenvectors associated with these eigenvalues can be used for clustering. 
\end{remark}

 The clustering properties of the eigenvectors of the matrix $\mathbf{P}$ can be most clearly understood in an ideal scenario. Consider a graph $G$ with nodes $\Omega$ and edge weights defined by $\mathbf{P}$, and suppose the graph is divided into $\bar{k}$ connected components $G_i $ with $i = 1\ldots \bar{k}$. If all elements within each component have identical edge weights connecting them, the normalized affinity matrix $\textbf{P}$ would be a block diagonal matrix. This matrix would have $\bar{k}$ non-zero eigenvalues $\{\lambda_i\}_{i = 1}^{\bar{k}}$, each equal to 1, and the corresponding eigenvectors $\{u_i\}_{i = 1}^{\bar{k}}$ would be piecewise constant indicator functions corresponding to these components, indicating the correct partitions. 

In real-world applications, however, the affinity matrices are not perfectly block-diagonal and are typically perturbed versions of the ideal matrices. 
Thus, the spectrum of $\mathbf{P}$ is not perfectly dichotomous. The goal is to approximate the ideal case, where points within each cluster are strongly connected to their neighbors and weakly connected to points in other clusters. Determining the number of significant eigenvectors (i.e., choosing $\bar{k}$) is crucial. Traditional methods include finding the maximum eigengap or minimizing a cost function (\cite{ZLP04}). It is possible to opt, as proposed in \cite{BCCS14}, for a semi-supervised approach by setting a significance threshold $\varepsilon$, selecting eigenvectors where $\lambda_i > 1-\varepsilon$. Given the sensitivity of results to changes in $\varepsilon$, a technique from diffusion maps (\cite{coifman2005geometric, CL2006, DHHV08}) is applied, using an auxiliary parameter $\tau$ to evaluate the exponentiated spectrum $\{\lambda_i^\tau\}_{i = 1}^k$. For large $\tau$, this spectrum better approximates a dichotomous form, facilitating more stable clustering. After determining the number $\bar{k}$ of significant eigenvectors to use, a basic clustering technique (e.g., see \cite{kannan2004clusterings}) is employed to extract clustering information. 

\begin{remark}
Normalizations for connections between cells have been neurally studied for example in \cite{CH12, BBB09, TSE94, H92}. In particular, among the normalizations presented in \cite{CH12}, there is a normalization corresponding to \eqref{eq:norm} since it computes a ratio between the responses of an individual neuron and the summed activity of pooled neurons. This normalization was introduced to study the properties of neurons in the primary visual cortex \cite{H92}.
\end{remark}

\subsection{Solution of the  stereo correspondence problem}
\label{sec:reconstruction}
To use spectral clustering to solve the stereo problem, we start  by mapping the left and right images to $\PO$, pairing all possible matchable points. This is achieved by inverting perspective projections through the left and right optical centers $C_L = (-c, 0, 0)$ and $C_R = (c, 0, 0)$ onto parallel retinal planes. Specifically, points on the left and right retinal planes that share the same y-coordinate are paired.

Using classical triangulation techniques (e.g., see \cite{F93}), we start by considering any pair $Q_L = (x_L, y)$ and $Q_R = (x_R, y)$ with their respective orientations identified by the angles $\theta_L$ and $\theta_R$. Here, $\theta_L$ and $\theta_R$ represent the orientations of the tangent vectors to the visual stimulus at points $Q_L$ and $Q_R$ on the left and right retinal planes, respectively.
From any pair $(Q_L, Q_R)$, we project back into $\mathbb{R}^3$ to obtain the 3D coordinates of the point $Q = (r_1, r_2, r_3)$, where:
\begin{equation}
r_1 = c\frac{x_L + x_R}{x_L - x_R}, \qquad
r_2 = \frac{2cy}{x_L - x_R}, \qquad 
r_3 = \frac{2fc}{x_L - x_R},
\end{equation}
with $c > 0$ and $f$ being the focal length.

The corresponding tangent vector $t$ has direction identified by $n(\theta, \varphi):=(\cos\theta \sin\varphi, \sin\theta \sin\varphi, \cos\varphi)$, where the angles $\theta$ and $\varphi$ are computed as follows (e.g., see \cite{F93}).  First, we compute the projection matrices $\Pi_L$ and $\Pi_R$, as well as the vectors $q_L, q_R, t_L, t_R$. These are given by:

\begin{equation}
\Pi_L = \begin{bmatrix}
1 & 0 & -\frac{c}{f} \\
0 & 1 & 0 \\
0 & 0 & 1
\end{bmatrix}, 
\Pi_R = \begin{bmatrix}
1 & 0 & \frac{c}{f} \\
0 & 1 & 0 \\
0 & 0 & 1
\end{bmatrix}, \\
q_i = \begin{pmatrix}
x_i \\
y_i \\
f
\end{pmatrix},
t_i = \begin{pmatrix}
\cos(\theta_i) \\
\sin(\theta_i) \\
0
\end{pmatrix} \text{ for } i = L, R.
\end{equation}
Next, we compute the tangent vector $t$ using the cross product $\times$ of the transformed vectors:
\begin{equation}
t = (\Pi_L^{-1}q_L\times \Pi_L^{-1}t_L )\times (\Pi_R^{-1}q_R\times \Pi_R^{-1}t_R).
\end{equation}
Finally, we calculate the angles $\theta$ and $\varphi$ as:
\begin{equation}
\theta = \arctan(t_2/t_1) , \quad \varphi = \arccos(t_3/\|t\|),
\end{equation}
where $t_i$ denotes the ith component of the vector $t$.

This process enables us to reconstruct the point cloud in $\PO$, representing the true stimulus plus some noise deriving from false matches.
To address the correspondence problem and identify the correct points belonging to the real stimulus, we perform a spectral analysis; i.e., we compute the affinity matrix $\mathbf{J}$ for each pair of points in $\PO$ and consider its spectral decomposition. We use the eigenvectors to capture the most significant part of the visual scene, corresponding to the real stimulus. This selection of objects in 3D determines the matching, proposing an alternative resolution to the stereo correspondence problem.

\section{Numerical simulations and results}

In this section, we detail the implementation of our proposed model for identifying 3D perceptual units and solving the stereo correspondence problem. We then discuss the results obtained.

\subsection{Numerical solution of the stochastic differential equation}
We begin by outlining the numerical method used to compute the connectivity kernel. This is followed by an error analysis and a review of the numerical results.

\subsubsection{Monte Carlo implementation}
\label{sssec:monte_carlo}
The Monte Carlo method is widely used due to its effectiveness in various situations. Its implementation consists of two main steps: first, simulating stochastic paths using the Euler-Maruyama scheme, and then taking an appropriate average of the simulated paths in accordance with the Strong Law of Large Numbers.

\paragraph{Euler-Maruyama scheme.}

We fix the time parameter $T>0$ and the number of steps $M \in \N$ of the path; then, we construct, iteratively, the discrete version of the stochastic process $\Gamma_t = (r_1(t), r_2(t), r_3(t), \varphi(t),\theta(t)) $ in $\PO$, satisfying equation \eqref{SDE}.  

The scheme generates
 $\Gamma_{kT/M}^M$ with $0 \leq k \leq M-1 $, by setting 
\begin{equation}
\Gamma_{(k+1)T/M}^M = \Gamma_{kT/M}^M+Y_3(\Gamma_{kT/M}^M)\frac{T}{M} + \varlambda \sigma(\Gamma_{kT/M}^M)\sqrt{\frac{T}{M}}\delta_k
\end{equation}
with $\Gamma_0^M = \Gamma_0$ and $\sqrt{\frac{T}{M}}\delta_k$ comes from the fact that the Brownian increments $d B_i(t)$ follow the Gaussian law $N(0, \frac{T}{M})$, with zero mean and variance $\frac{T}{M}$, so that $\delta_k = (\delta_1^k, \delta_2^k)$ and $\delta_i^k, i = 1, 2$ has normal distribution $N(0,1)$. In coordinates, we have the following system: 
\begin{equation}
\label{randomPath}
\begin{cases}
r_1(k+1) = r_1(k) + \frac{T}{M}\cos\theta(k) \sin \varphi(k) \\
r_2(k+1) =r_2(k) + \frac{T}{M}\sin\theta(k)\sin\varphi(k) \\
r_3(k+1) =r_3(k) +\frac{T}{M}\cos\varphi(k) \\
\theta(k+1) =\theta(k)- \varlambda\sqrt{\frac{T}{M}}\frac{\delta_1^k}{\sin\varphi(k)}\\
\varphi(k+1)= \varphi(k)+ \varlambda\sqrt{\frac{T}{M}}\delta_2^k \\
\end{cases} \text{ for } k = 0, 1, \ldots, M-1,
\end{equation}
with initial point $\Gamma_0 = (r_1(0), r_2(0), r_3(0), \varphi(0), \theta(0))$ and $\varlambda \in \R$ diffusion parameter. 
This very simple method allows a convergence rate result very satisfying for $M \gg 1$. Further details can be found in  \cite{GT13}. 
 
\begin{remark}
A similar approach was introduced in \cite{DBM19} to numerically compute the fundamental solution associated with the Kolmogorov equation of a diffusion process in position-orientation space. The proposed discretization \cite[eq. (72)]{DBM19} is based on the idea that a point in $\S^2$ can be represented in terms of three-dimensional rotations starting from an initial axis. 
 If we reformulate \eqref{randomPath} using rotations, we have: 
\begin{equation}
\begin{cases}
r(M) = r(0) + \sum_{k = 0}^{M-1} \frac{T}{M} n(k+1)\\
n(k+1) = \left( R_3^{\theta_{(k+1)}}R_2^{\varphi_{(k+1)}}\right)e_3 
\end{cases} \text{with }\begin{cases}
\theta(k+1) =\theta(k)- \varlambda\sqrt{\frac{T}{M}}\frac{\delta_1^{k}}{\sin\varphi(k)}\\
\varphi(k+1)= \varphi(k)+ \varlambda\sqrt{\frac{T}{M}}\delta_2^k \\
\end{cases}
\end{equation}
and set $n(k+1) := (\sin\theta_{(k+1)}\cos\varphi_{(k+1)}, \cos\theta_{(k+1)}\sin\varphi_{(k+1)}, \cos\varphi_{(k+1)})^T\in \S^2$, and the point $r(k) := (r_1(k), r_2(k), r_3(k))^T \in \R^3$ for $k = 0 \ldots M-1$.  Note: this approach differs from the discretization proposed in \cite[eq. (72)]{DBM19}, where the spherical term $n(k+1)$ is computed as a consecutive product of rotations using the natural index $k\leq M-1$.
\end{remark}

\paragraph{Strong Law of Large Numbers.}

We denote with $\Gamma_{(t, \xi_0, t_0)}$ the stochastic process at time $t$ associated to \eqref{SDE} starting from the point $\xi_0$ at time $t_0$. 
Its density law $\rho_\varlambda$  is defined through a probability measure $P$ as 
\begin{equation}
\rho_\varlambda(\mathcal{S}, t; \xi_0, t_0) := P[\Gamma_{(t, \xi_0, t_0)}\in \mathcal{S}] = \mathbb{E}[\mathbbm{1}_\mathcal{S}(\Gamma_{(t,\xi_0, t_0)})],
\end{equation}
with $\mathbb{E}$ the average in probability and $\mathbbm{1}_\mathcal{S}$ indicator function over a subset $\mathcal{S}$ of $\PO$.
To recover the density from the approximated stochastic process \eqref{randomPath}, it is possible to use the Strong Law of Large Numbers.

\begin{theorem}{(Strong Law of Large Numbers)}
\label{thm:SLLN}
Let $\Gamma_t^{(i)}, i \geq 0$ be a sequence of independent and identically distributed random variables. Assume that $\mathbb{E}[\Gamma_t^{(i)}]< \infty.$ For $N \geq 1$, denote the empirical mean of $(\Gamma_t^{(1)}, \dots, \Gamma_t^{(N)})$ by $\tilde S_N = \frac{1}{N}\sum_{i=1}^N \Gamma_t^{(i)}$. Then, the Strong Law of Large Numbers holds true: 
\begin{equation}
\label{eq:SLLN}
\lim_{N\rightarrow \infty}\tilde S_N = \mathbb{E}[\Gamma_t^{(1)}], \hspace{ 0.5cm } P-a.s.
\end{equation}
\end{theorem}

To apply this theorem to our case, we first concentrate on the set $\mathcal{S}$ of interest.  We fix a discrete covering grid $\{\xi_j\}_{j \in \N, j\leq J}$ of a suitable subset $V$ of $\PO$: this is a collection of subsets $\{\Omega_i\}$ satisfying $\Omega_i\cap \Omega_j = \emptyset$ if $i\neq j$ and $ \cup_{j=1}^{J}\Omega_j= V$.
If we perform a discretization on $\R^3$ with step size $\Delta_1, \Delta_2, \Delta_3$ and on $\S^2$ with step size $\Delta_\varphi, \Delta_\theta$, we assign the element $\xi_j$ to be a representative of the $j^{th}$ box. In this sense we define $\mathcal{S} := \{\xi_j\}_{j \in \N, j \leq J}$. 

Then,  from the discrete counterpart of \eqref{eq:SLLN} we get: \begin{equation}
\label{transitioPb}
\rho_\varlambda(\xi_j, t; \xi_0, t_0) = \frac{1}{N} \sum_{i = 1}^N \mathbbm{1}_{\xi_j}(\Gamma^{(i)}_{(t, \xi_0, t_0)}),
\end{equation}
considering as stochastic process $\mathbbm{1}_{\xi_j}(\Gamma^{(i)}_{(t, \xi_0, t_0)})$. 

\begin{figure}[tbh]
     \centering
  \begin{subfigure}[b]{0.45\textwidth}
         \centering
         \includegraphics[width=\textwidth]{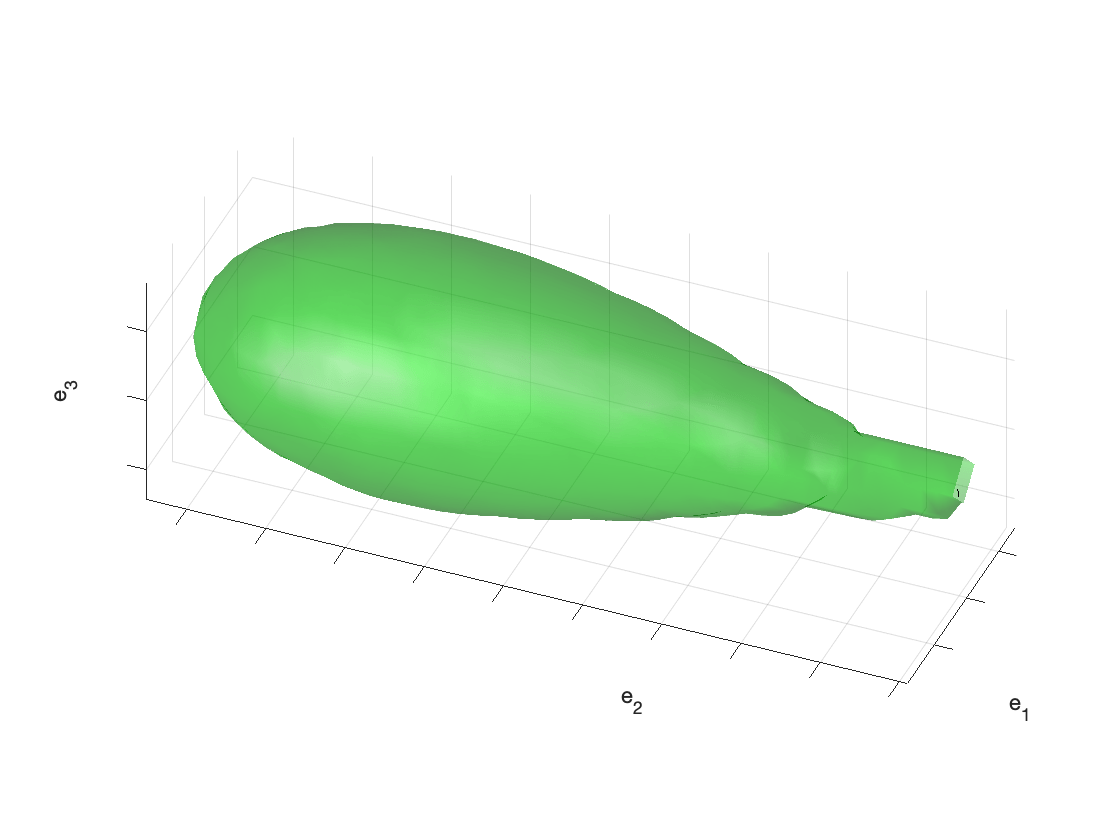}
         \caption{$J_{\R^3}^{\varlambda = 0.035}(\xi_j, \xi_0)$}
     \end{subfigure}
             \centering
     \begin{subfigure}[b]{0.45\textwidth}
         \centering
         \includegraphics[width=\textwidth]{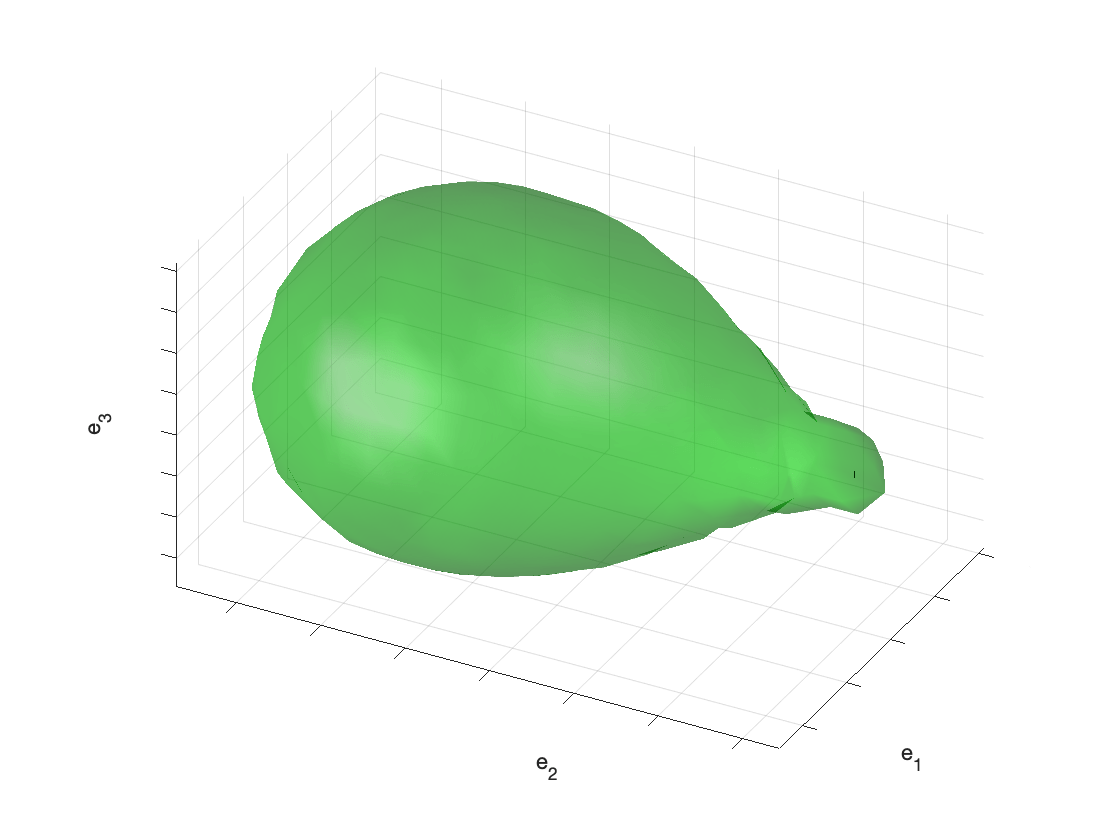}
         \caption{$J_{\R^3}^{\varlambda = 0.08}(\xi_j, \xi_0)$}
     \end{subfigure}

           \caption{ Display of $J_{\R^3}(\xi_j, \xi_0)$  for two different diffusion coefficients $\varlambda$ (left $\varlambda = 0.035$, right $\varlambda = 0.08$), at isovalue $0.1$. Euler-Maruyama scheme parameters are $M = 400$, $T = 100$, $N = 10^6$.
           }
\label{kernel}
\end{figure}

In other words, for a given $\xi_0 \in \PO$ we simulate $N$ discrete-time random paths and assign to each region $\xi_j$ a value between 0 and 1 corresponding to the number of paths that passed through it at the final time $t$ divided by $N$. This provides a distribution over the cells $\xi_j$ that, up to a multiplicative constant, for large values of $N$ gives a discrete approximation of the kernel. 

What we have just obtained is a probability density relative to the Kolmogorov equation associated to the operator \eqref{fKoperator}. Since a time integration is sufficient to obtain the time-independent fundamental solution, in the discrete case this corresponds to a sum over $t$:
\begin{equation}
\label{eq:kernel_sim}
J_\varlambda^T(\xi_j, \xi_0) = \sum_{t = t_0}^T \rho_\varlambda(\xi_j, t; \xi_0, t_0).
\end{equation}
This equation corresponds to applying the numerical method that generates \eqref{transitioPb} on all the points of the random path \eqref{randomPath} (relative to different evolutions times $t< T$), obtaining an approximation for $J_\varlambda ^T$.

Examples of the numerical kernels are shown in Figure \ref{kernel}.
The iso-surfaces in $\R^3$, images (a) and (b), are obtained considering the marginal distribution
\begin{equation}
\label{eq:intensityR3}
J_{\R^3}(\xi_j, \xi_0) = \int_{\S^2}J_\varlambda^T(\xi_j, \xi_0)d\sigma, 
\end{equation} where $d\sigma = \sin\varphi d\theta d \varphi$ is the spherical measure on $\S^2$.

\paragraph{Error estimate.}

The Monte Carlo method is based on the Strong Law of Large Numbers and clearly approximates results whose accuracy depends on the number of values $N$ used in equation \eqref{eq:SLLN}.

We display error bars for some sections of our kernel, according to the following procedure. Keeping fixed the initial point $\xi_0 \in \PO$, and integrating over $\S^2$, analogously as in \eqref{eq:intensityR3}, the kernel can be viewed as a function of the variable  $(r_1, r_2, r_3) \in \R^3$, namely $J_\varlambda^T = J_\varlambda^T(r_1, r_2, r_3)$. 
For simplicity, we focus on a section of the kernel: for example, fixing $(\bar r_2, \bar r_3)$ we consider $J_\varlambda^T(r_1)= J_\varlambda^T(r_1, \bar r_2, \bar r_3)$. We estimate the variance $\bar \sigma$ through the equation 
\begin{equation}
\label{eq:variance}
\bar \sigma^2 = \frac{1}{N}\sum_{i = 1}^N \left(\mathbbm{1}_{\xi}\left(\Gamma^i\right)-J_\varlambda^T\right)^2,
\end{equation}
and an error bar corresponding to the confidence interval $I = [J_\varlambda^T-r, J_\varlambda^T+r]$ where $ r = \frac{2.57 \bar \sigma}{\sqrt{N}}$ at the point of coordinates $(r_1, J_\varlambda^T(r_1))$ is then placed along the vertical axis. The same procedure can be applied to the other coordinates. We refer to \cite{N14} for further information.

Figure \ref{fig:error_bars} shows error bars for $J_\varlambda^T(r_1)$, namely a section parallel to the $r_1-$axis, related to kernels simulated with a different number of paths.
Image (a) shows the comparison between the values of a kernel generated by a number of paths $N = 10^5$ with one generated considering $N = 10^6$, in a neighborhood of the pole $\xi_0$: the magnitude of the confidence interval, indicated in the picture as $\Delta(N)$, decreases as $N$ increases. Image (b) shows that, as we get further and further away from $\xi_0$ (both along the $r_1$-axis and on the $r_2$-axis), the kernel value decays to zero, as does the standard deviation. This is well in agreement with Remark \ref{rmk:gauss_est}: as time becomes larger, the kernel value tends to zero.

Figure \ref{fig:error_perc} shows the error along the section parallel to the $r_2-$axis and passing through $\xi_0$. This section is orthogonal to the ones studied in Figure \ref{fig:error_bars}. Image (a) displays the trend of the error bars on the values $J_\varlambda^T(r_2)$: moving along the $r_2-$axis the function decreases to zero, as does the error.

\begin{figure}[H]
     \centering
  \begin{subfigure}[b]{0.4\textwidth}
         \centering
         \includegraphics[width=0.9\textwidth]{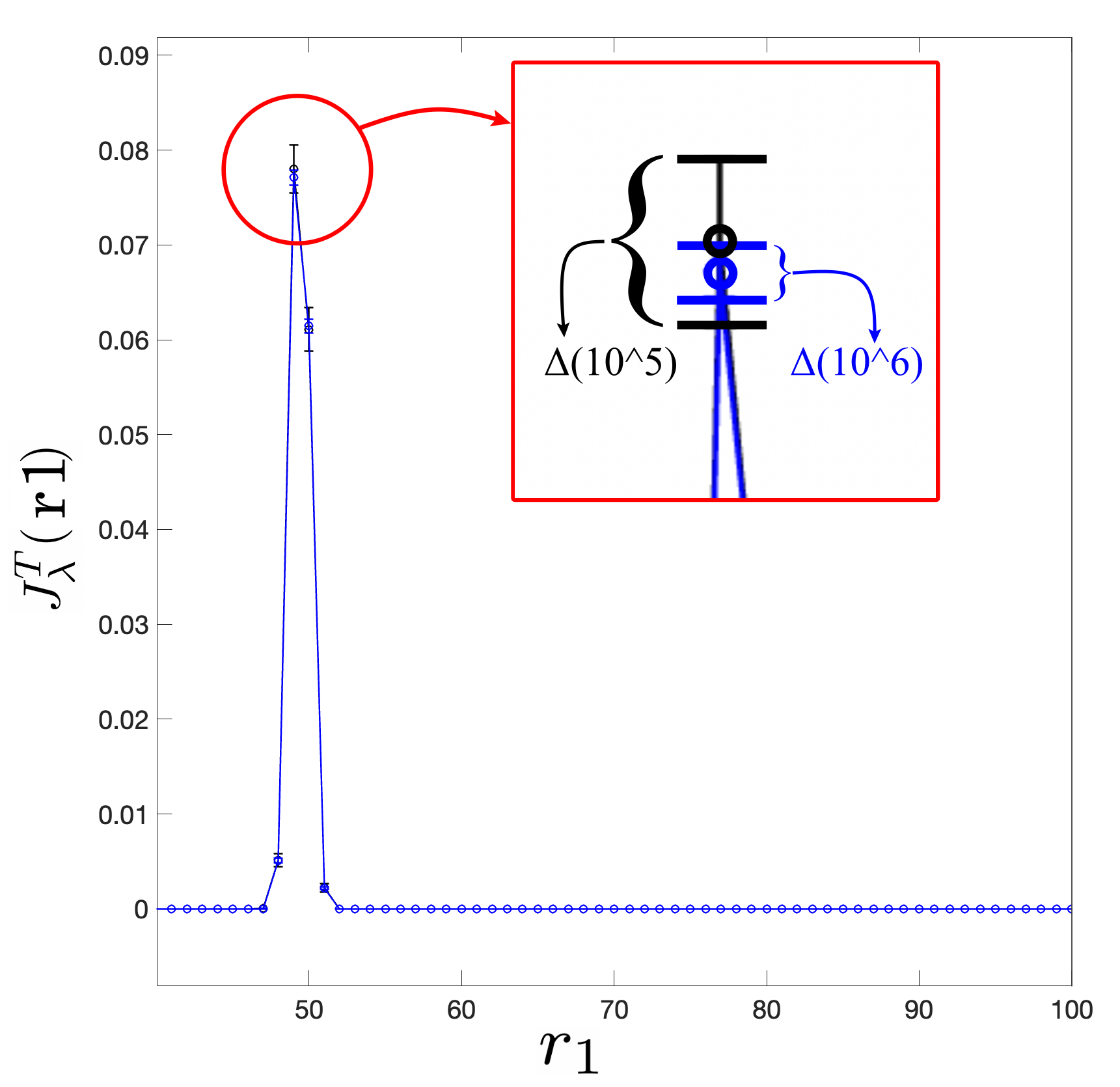}
         \caption{}
     \end{subfigure}
                  \centering
     \begin{subfigure}[b]{0.55\textwidth}
         \centering
         \includegraphics[width=\textwidth, height = 5cm]{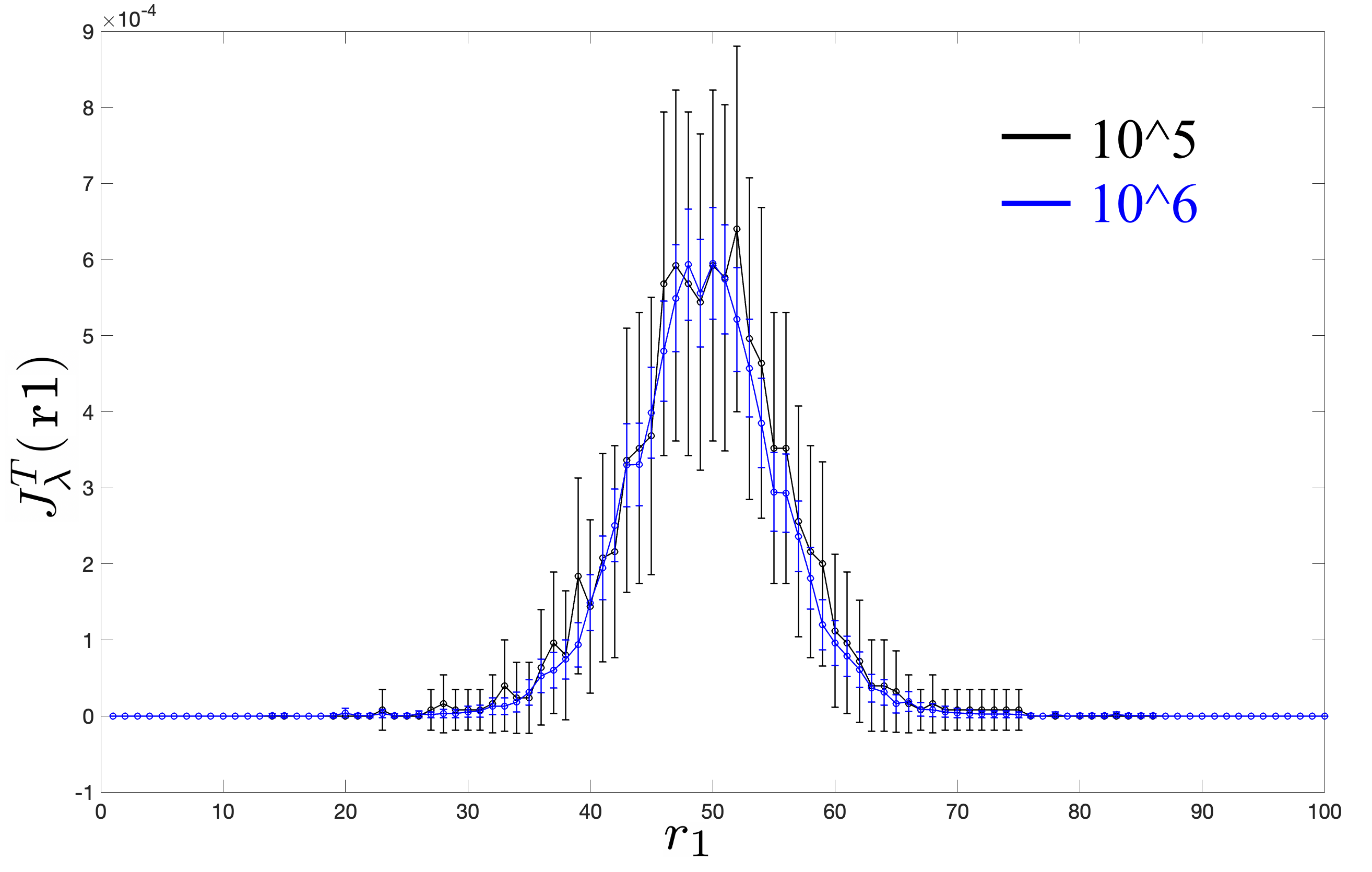}
         \caption{}
     \end{subfigure}
           \caption{Display of error bars on $J_\varlambda^T(r_1)$ sections with $\varlambda = 0.035$ and $T= 100$.  The initial point $\xi_0$ is described by  spatial indices $(r_{1_0},r_{2_0},r_{3_0})=(50, 1, 50)$, while $(\theta_0, \varphi_0)=(\pi/2, \pi/2)$.
 (a) Display of $J_\varlambda^T (r_1)$ identified by $(\bar r_2, \bar r_3) = (10, 50) $  for two different numbers of paths $N$: blue color corresponds to $N= 10^6$, the black one to $N = 10^5$. Emphasized in the red square: the amplitude ($\Delta(N)$) of the two confidence intervals. The width of the blue interval is about half of the black one.
  (b) $J_\varlambda^T (r_1)$  identified by $(\bar r_2, \bar r_3) = (50, 50) $ for two different number of paths: blue correspond to $N= 10^6$, black to $N = 10^5$. Moving away from the pole, both the kernel value and the confidence interval decrease.}
\label{fig:error_bars}
\end{figure}

\begin{figure}[tbh]
     \centering
  \begin{subfigure}[b]{0.5\textwidth}
         \centering
         \includegraphics[width=\textwidth]{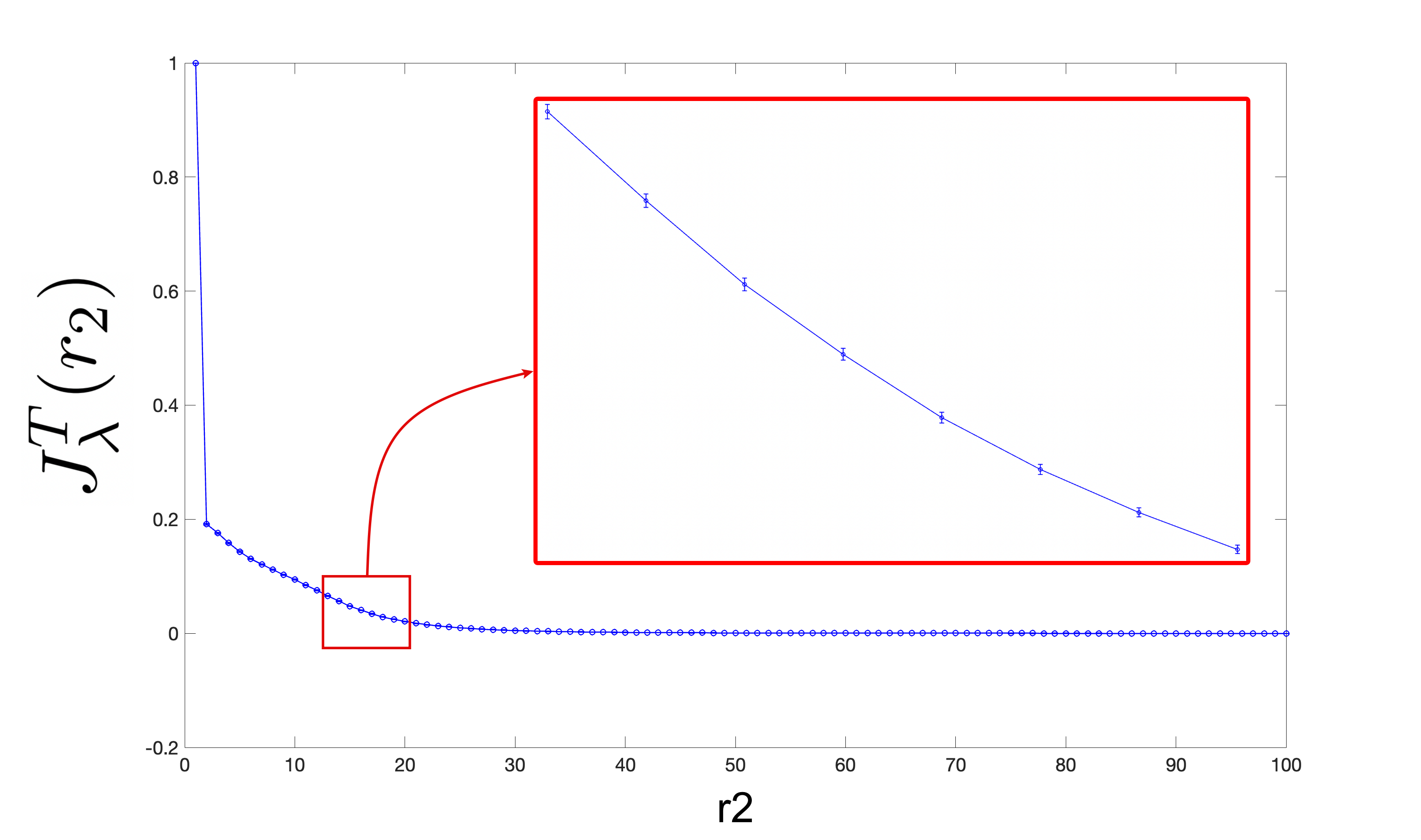}
         \caption{}
     \end{subfigure}
                  \centering
     \begin{subfigure}[b]{0.4\textwidth}
         \centering
         \includegraphics[width=\textwidth]{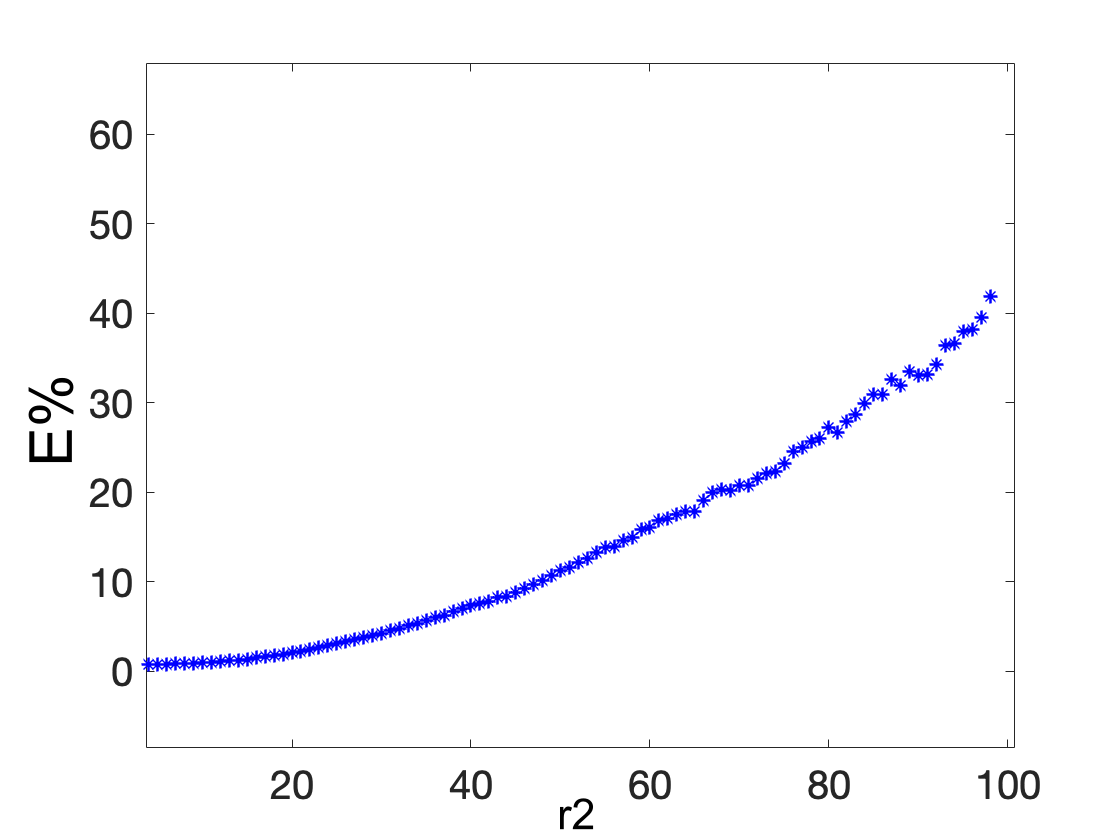}
         \caption{}
     \end{subfigure}
           \caption{Display of error bars and percentage error for $J_\varlambda^T(r_2)$, with kernel parameters $\varlambda = 0.035$ and $T= 100$, $N = 10^6$.  The initial point $\xi_0$ is described by $(r_{1_0},r_{2_0},r_{3_0})=(50, 1, 50)$  and $(\theta_0, \varphi_0)=(\pi/2, \pi/2)$. The section has indices $(r_{1_0},r_{3_0})=(50, 50)$.
           (a) Display of error bars on $J_\varlambda^T(r_2)$. (b) Display of the percentage error $E_{\%}=100\frac{r}{J_\varlambda^T(r_2) }$, with $ r = \frac{2.57 \bar\sigma }{\sqrt{N}}$.   }
\label{fig:error_perc}
\end{figure}

 This is well in accordance with results of Figure \ref{fig:error_bars} and Remark \ref{rmk:gauss_est}. On the other hand, image (b) of Figure \ref{fig:error_perc} shows the behavior of the percent error, defined here as $E_{\%}:=100\frac{r}{J_\varlambda^T(r_2)}$, with $ r = \frac{2.57 \bar\sigma }{\sqrt{N}}$.  In this case, the tendency is the opposite compared to image (a): we have a good percentage estimate near $\xi_0$, while this goodness decreases as the distance from the $\xi_0$ increases. This happens because near the pole the points are reached by a high number of stochastic paths, and this number decreases as we move away from $\xi_0$.

\subsubsection{Dependence on parameters}
The kernel defined by Monte-Carlo approximation mainly depends on $ 4 $ parameters: $ M $ the number of steps of the simulated path, $ N $ the number of paths considered, $ \varlambda $ the diffusion coefficient, and $ T $ the final evolution time of the generated path. We will always consider the first two parameters quite large to ensure the convergence of the method; the other two parameters affect the behavior of the kernel, as can be seen in Figure \ref{scala}.
\begin{figure}[H]
\begin{center}
\includegraphics[width =0.9\textwidth ]{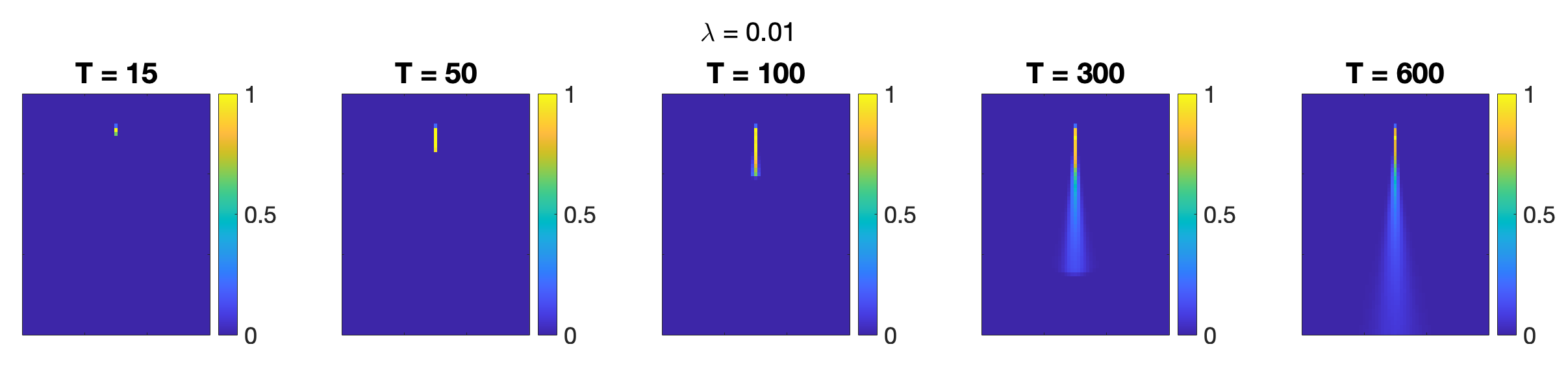}
\includegraphics[width =0.9\textwidth ]{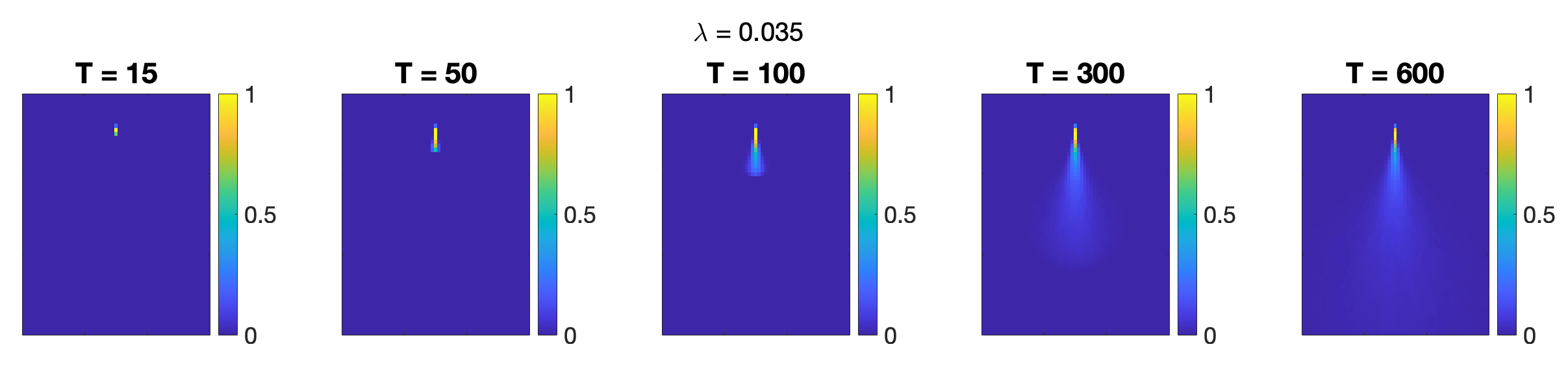}
\end{center}
\caption{Marginal projections on the plane $(r_1, r_2)$ to display dependence on the parameters $\varlambda$ and $T$. The abscissa corresponds to $r_1$- axis, while the ordinate to $r_2$-axis. Columns describe the scale parameter $T \in \{ 15, 50, 100, 300, 600\}$, while rows correspond to different values of diffusion $\varlambda \in \{ 0.01, 0.035\}$. }
\label{scala}
\end{figure}

This diffusion parameter operates a modification on the thickness of the kernel: the bigger $\varlambda$ is, the thicker the kernel, making the diffusion terms prevail. On the other hand, the smaller $\varlambda$ is, the thinner the kernel, mainly concentrated on the initial direction $n(\varphi_0, \theta_0)$, making the transport term the leading term characterizing the equation.

On the other hand, the temporal parameter  $ T $ can be seen as a scale parameter.
The effect of the variation of the parameter $ T $ on the shape of the kernel is shown in Figure \ref{scala}: the columns show a proportional relationship between the increase in the final time of the stochastic path, and the amplitude of the kernel. 
During numerical experiments for visual grouping, this parameter will be taken in accordance with the image dimension. 

\subsubsection{Strenght of correlations}

We define kernels through transition probabilities (equation \eqref{eq:kernel_sim}) to express correlation rates between neighboring points. Figure \ref{compF} illustrates the correlation of position-orientation elements with a fixed starting point $p \in \mathbb{R}^3$ and orientation $(\theta, \varphi) = (\pi/2, \pi/2)$ in 3D Cartesian space. 
\begin{figure}[tbh]
     \centering
         \includegraphics[width=0.5\textwidth]{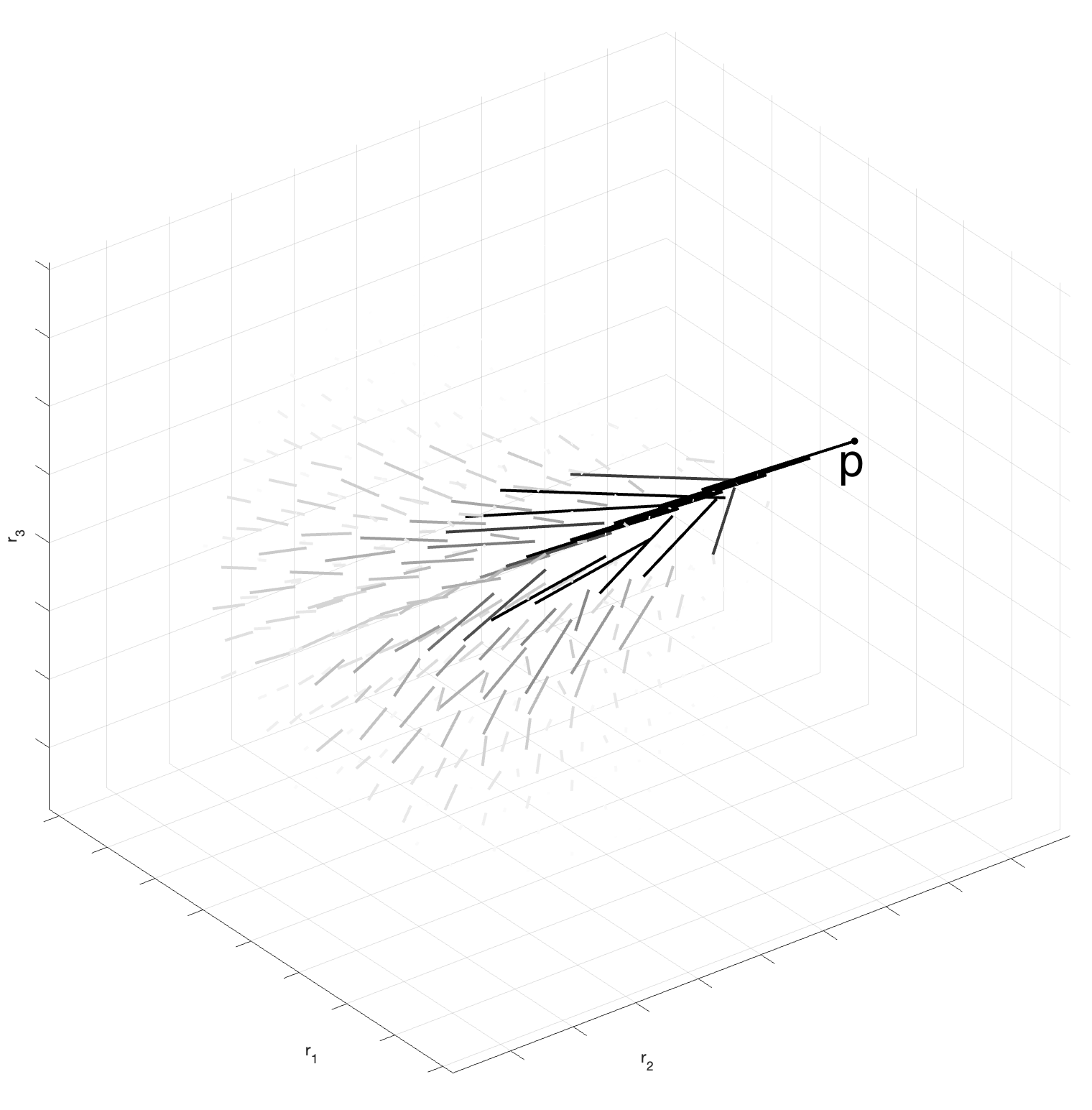}
           \caption{Position-orientation elements correlated with initial point $p \in \R^3$ and initial direction $(\theta, \varphi) = (\pi/2, \pi/2)$ in the first coordinate chart. Intensity decay depends on position and orientation: the darker and longer the segment is, the more it correlates with the starting point.}
\label{compF}
\end{figure}
Positions and orientations in 3D space influence compatibility. Figure \ref{compF} shows position-orientation elements selected by the kernel, each associated with a preferred orientation and intensity computed from equation \eqref{eq:kernel_sim}\footnote{For each position in 3D, we show the most likely orientation. This likelihood it is then used as intensity.}. The intensity is represented by the color and length of the 3D segment, with darker and longer segments indicating stronger correlations. These measurements encode bending and twisting information, as position-orientation elements are not restricted to frontal planes. We note the intensity to decrease with distance from the initial position, demonstrating clear position dependence; the same pattern also holds for orientations compatible with the initial direction.

\subsection{Grouping results}

In this section, we employ spectral analysis to identify 3D perceptual units and address the correspondence problem. We start by detailing the steps of the algorithm adapting the spectral clustering proposed in \cite{BCCS14, SC15} to our case of study. 

\subsubsection{The algorithm}
\label{sec:algorithm}
Building on the spectral analysis techniques described in Section \ref{sec:percepts-and-stereo}, we propose the algorithm described in the Algorithm box ``Spectral Clustering for Stereo Correspondence''. 
\begin{algorithm*}
\label{alg}
\caption*{\textbf{Spectral Clustering for Stereo Correspondence (SCSC)}}
\begin{algorithmic}[1]
\State  Recover the domain $\Omega \subset \PO$, ~ $\xi_i \in \Omega,  ~ i = 1, \ldots k$, from the coupling of retinal (rectified) images by inverting perspective projections.
\State Construct the affinity matrix $\mathbf{J}$ upon the appropriate connectivity measure $J$ of \eqref{eq:kernel_sim}:  $\textbf{J}_{ij}:=  J(\xi_i, \xi_j)$.
\State Compute the normalized affinity matrix $\textbf{P} =  \textbf{D}^{-1}\textbf{J}$. 
\State Solve the eigenvalue problem $\textbf{P}u_i = \lambda_i u_i$, where the indices $i = 1, \ldots, k$ are ordered such that $\{\lambda_i\}_{i = 1}^k$ are in decreasing order. 
\State Set a threshold $\varepsilon$ and a diffusion parameter $\tau$. 
\State Define $\bar{k} = \bar{k}(P; \varepsilon, \tau)$ as the largest integer such that $\lambda_i^\tau > 1 - \varepsilon$.
\State Create pre-clusters $ \{ \bar{C}_\ell \}_{\ell= 1}^{\bar{k}}$ where each point $\xi_i$ is assigned to the pre-cluster  $\bar{C}_\ell$ with:
\[ \ell = \operatorname{argmax}_{j = 1, \ldots, \bar{k}}u_j(i).\]
\State Fix a minimum cluster size $Q$, merge pre-clusters with fewer than $Q$ elements into a single cluster $C_0$, and order the remaining pre-clusters to obtain a final partition $\{C_\ell\}_{\ell = 1}^{K}$ where $K = K(P; \varepsilon, \tau, Q)$ and $K \leq \bar{k}$.
\end{algorithmic}
\end{algorithm*}

A distinctive feature of this spectral clustering approach is its adaptability in identifying meaningful clusters. Rather than pre-defining the number of initial clusters, the algorithm dynamically adjusts the number of pre-clusters based on the intrinsic properties of the data (\texttt{SCSC steps 5-7}). This flexibility allows the algorithm to more accurately capture the underlying structure of the data.

In the final clustering step (\texttt{SCSC step 8}), pre-clusters with fewer than a specified minimum number of elements, $Q$, are merged into a single cluster, $C_0$. This cluster $C_0$ is specifically designed to aggregate false matches resulting from the 3D reconstruction process. These false matches are points that do not conform to the perceptual law of good continuation, which typically ensures that coherent structures follow smooth, continuous paths. Essentially, $C_0$ represents elements that are not perceptually significant (e.g., see \cite{BCCS14}) and arises from incorrect pairings, thus should be excluded from the final coherent segmentation of the visual scene. The remaining clusters are tasked with effectively segmenting the visual scene into meaningful perceptual units. 

\subsubsection{Synthetic stimuli}

We first evaluate the algorithm using synthetic rectified stereo images (in this context, we assume that the positions and orientations in the retinal planes are available). The process involves matching all potential corresponding points from the left and right retinal images that share the same abscissa coordinate. We then determine the 3D positions and orientations by inverting the perspective projections and computing tangents, as described in Section \ref{sec:reconstruction}. The resulting lifted set $\Omega \subset \PO$ will consist of both elements present in the visual scene and some false matches, generally referred to as ``noise''. Spectral analysis is then applied to this dataset to assess the performance of the algorithm.

\paragraph{Parameterized curve.}
We begin by analyzing two retinal images, projection of a parameterized curve $\gamma:[0, T]\rightarrow \R^3$, which is represented by 30 discrete points. From this couple of bidimensional positions and orientations,  we perform the lifting of the stimulus into the 3D space $\PO$ (\texttt{SCSC step 1}). This process is depicted in Figure \ref{fig:stim-simple-curve}, where image (a) displays the left and right retinal images, and image (b) illustrates the subsequent 3D back-projection.

\begin{figure}[H]
     \centering
  \begin{subfigure}[b]{0.35\textwidth}
         \centering
         \includegraphics[width=\textwidth]{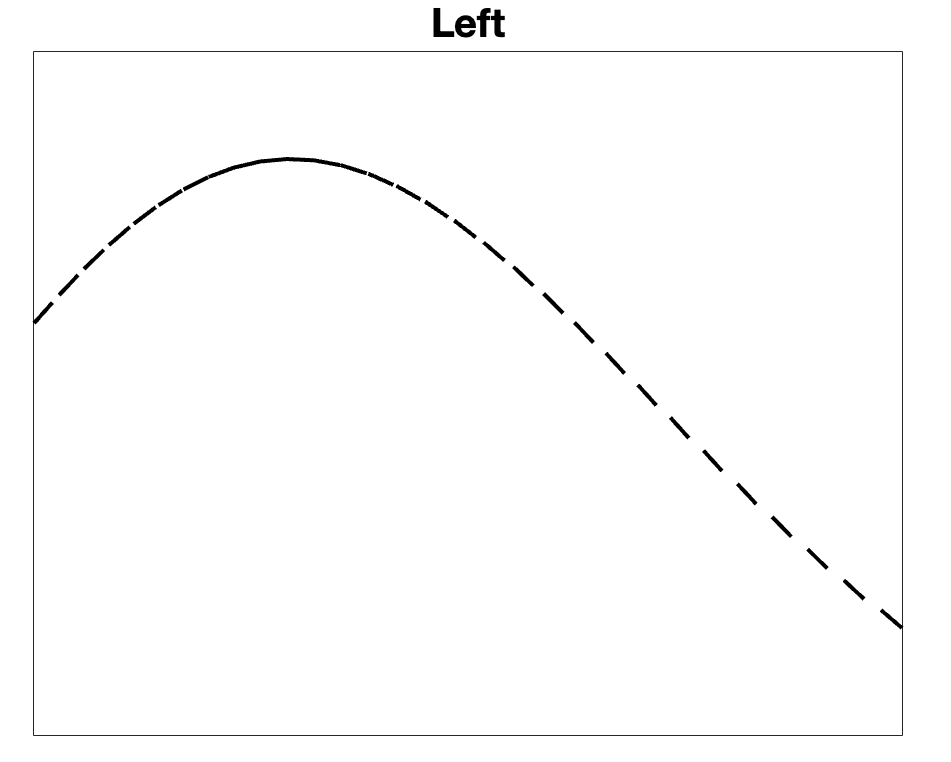}
         
         \includegraphics[width=\textwidth]{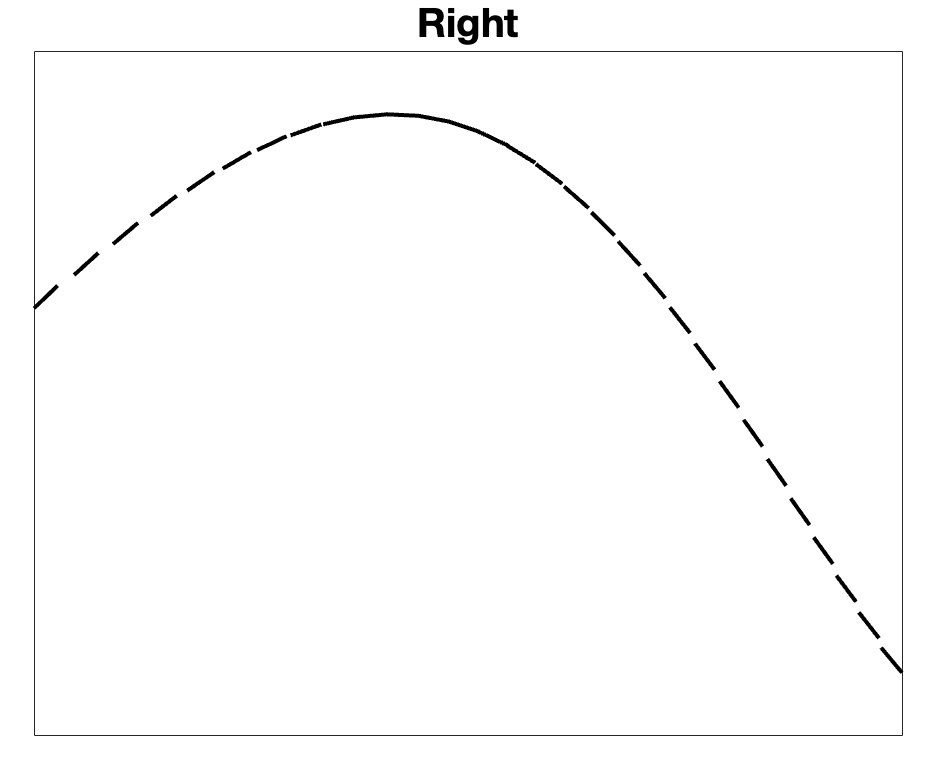}
         \caption{}
     \end{subfigure}
                  \centering
     \begin{subfigure}[b]{0.45\textwidth}
         \centering
         \includegraphics[width=\textwidth]{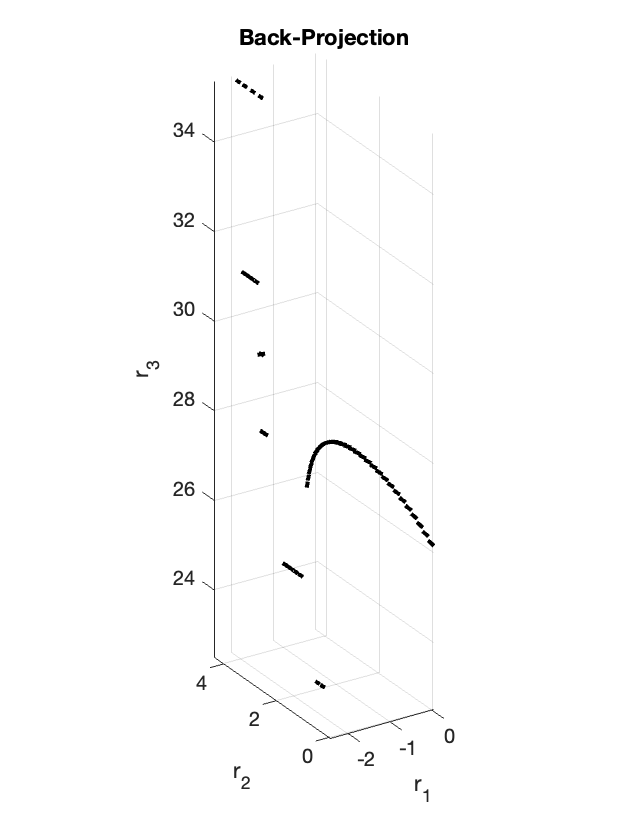}
         \caption{}
     \end{subfigure}
           \caption{Display of simple curve on retinal planes and its back-projection in 3D. (a) Left and right retinal images. (b) Lifting of the stimulus in $\PO$: the coupling generates the original 3D curve but also noise due to the false matching.}
\label{fig:stim-simple-curve}
\end{figure}

To address both the stereo correspondence problem and the segmentation of the three-dimensional scene, we apply our grouping algorithm. The results are illustrated in Figure \ref{1curvaGroup}. Image (a) shows the affinity matrix $\mathbf{P}$, calculated with kernel parameters for $\textbf{J}$ set to be $T= 95$ and $\varlambda = 0.0275 $ (\texttt{SCSC steps 2-3}). Image (b) presents the eigenvalues. The image on the left shows the original eigenvalues $\{\lambda_i\}_{i = 1}^k$ for $\mathbf{P}$, while the image on the right  displays the exponentiated eigenvalues $\{\lambda_i^\tau\}_{i = 1}^k$ for $\tau = 100$. The latter plot demonstrates a more distinct separation between significant and less significant eigenvalues, facilitating a more reliable determination of the number $\bar{k}$ of eigenvectors to use for clustering. In particular, colored in red are the eigenvalues that are greater than $1-\varepsilon$ with $\varepsilon = 0.01$. (\texttt{SCSC steps 4-6}). Pre-clusters (\texttt{SCSC step 7}) are illustrated in image (c), each color representing a different $\bar{C}_\ell$. The points are organized into three main pre-clusters: a blue cluster corresponding to the points on the 3D curve, an orange cluster with points sharing the same orientation but differing in spatial position, and a yellow cluster containing an outlier with a different orientation. Finally, image (d) shows the final clusters after applying a threshold of $Q = 25$ (\texttt{SCSC step 8}). Points in blue are those included in the single cluster that contains more than $Q= 25$ elements, while noise elements are shown in black and are assigned to cluster $C_0$ which is excluded from the segmentation and reconstruction processes. This  thresholding reveals that only one main cluster includes more than  $Q = 25$ points, which effectively recovers the true 3D curve. This outcome shows the algorithm's ability to accurately resolve stereo correspondence and segment the 3D scene.

\begin{figure}[H]
     \centering
     \begin{subfigure}[b]{0.3\textwidth}
         \centering
         \includegraphics[width=\textwidth]{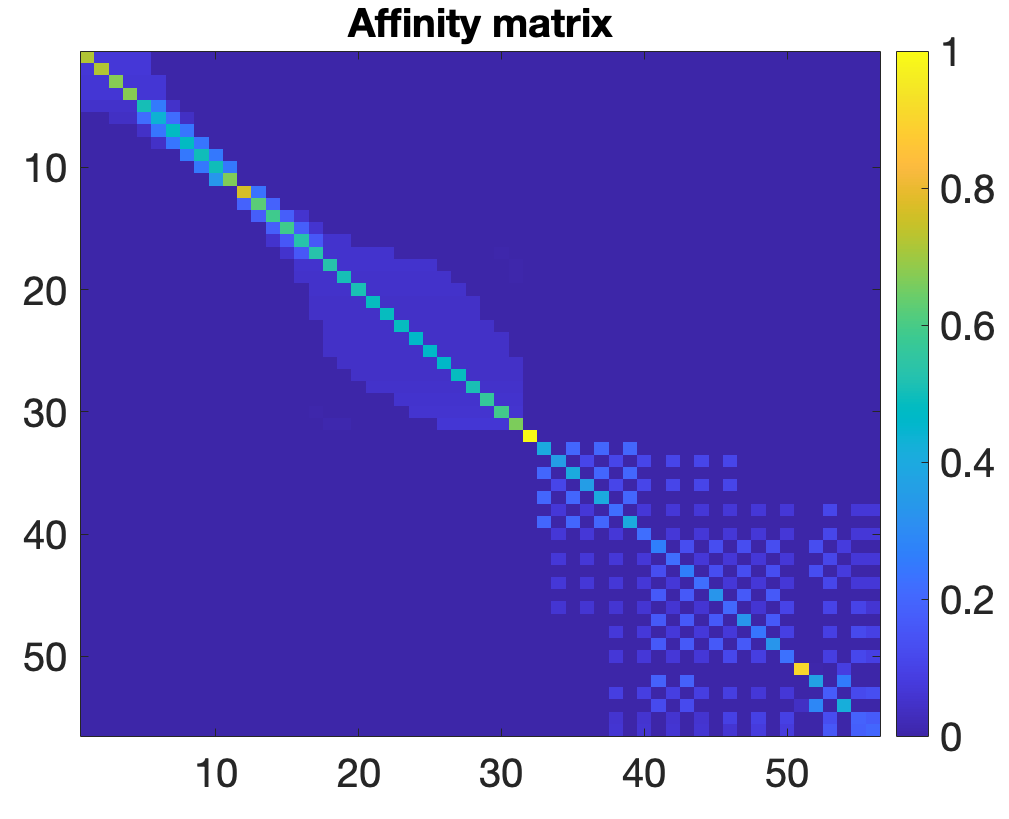}
         \caption{}
     \end{subfigure}
             \centering
     \begin{subfigure}[b]{0.65\textwidth}
         \centering
         \includegraphics[width=0.45\textwidth]{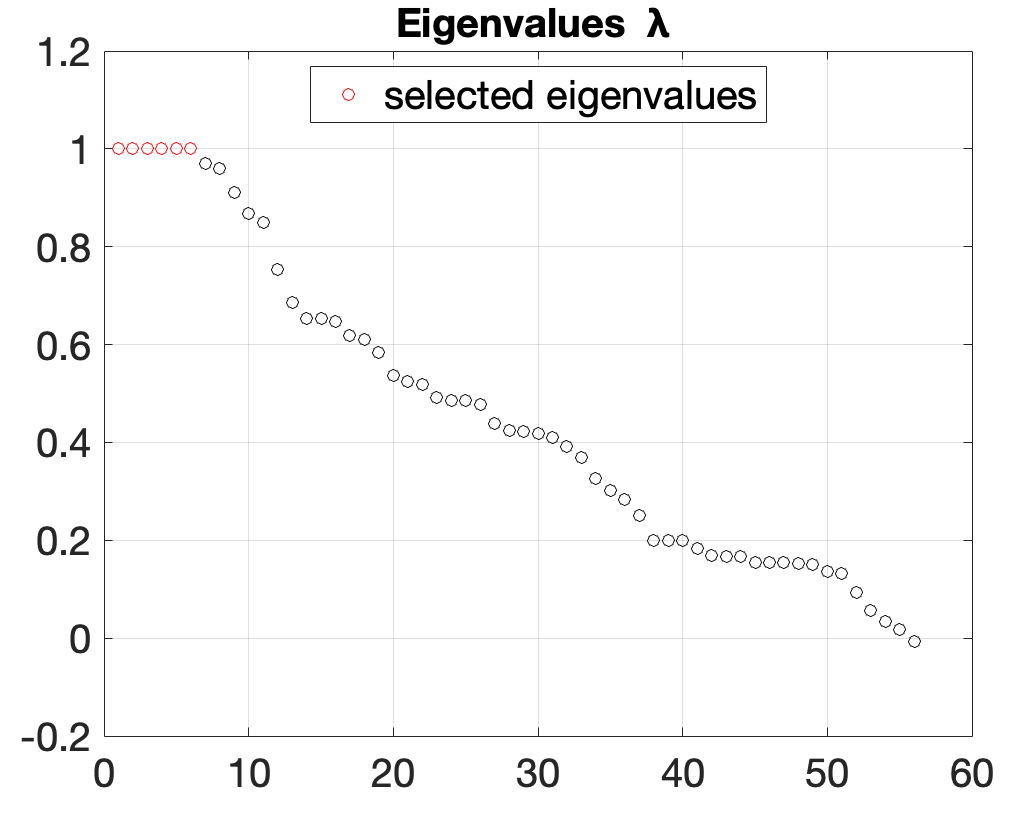}
         \includegraphics[width=0.45\textwidth]{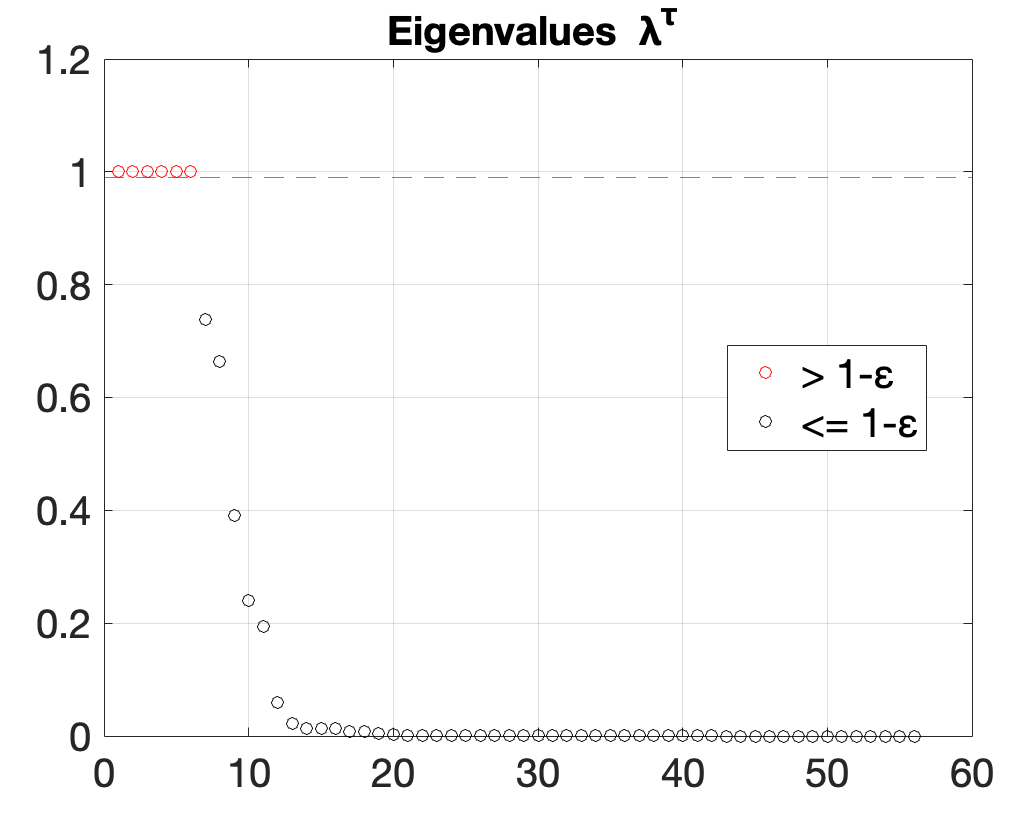}
         \caption{}
     \end{subfigure}
     
     \begin{subfigure}[b]{0.45\textwidth}
         \centering
         \includegraphics[width=\textwidth]{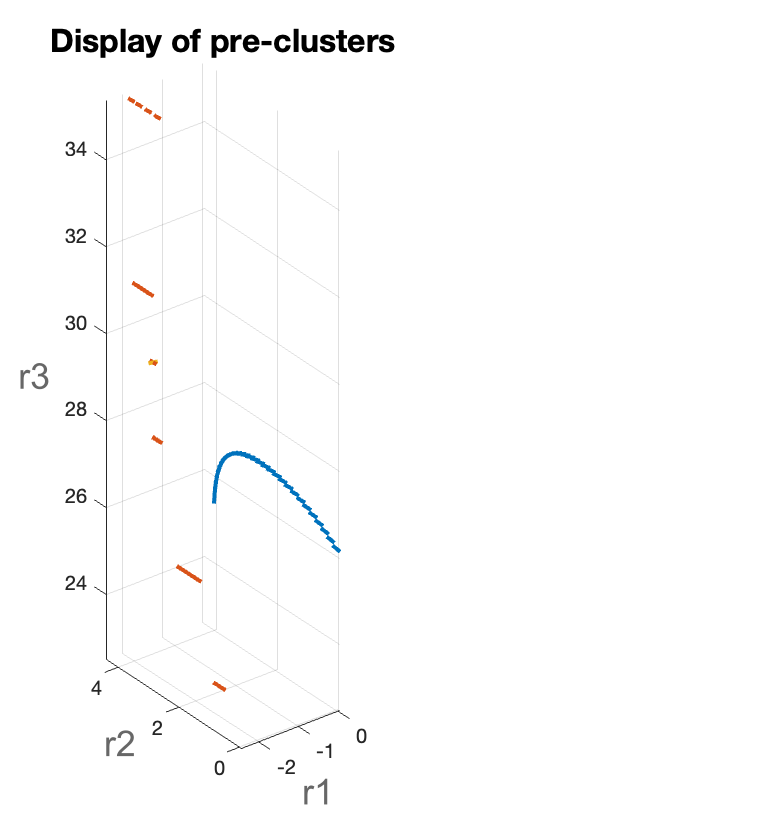}
         \caption{}
     \end{subfigure}
     \begin{subfigure}[b]{0.45\textwidth}
         \centering
         \includegraphics[width=\textwidth]{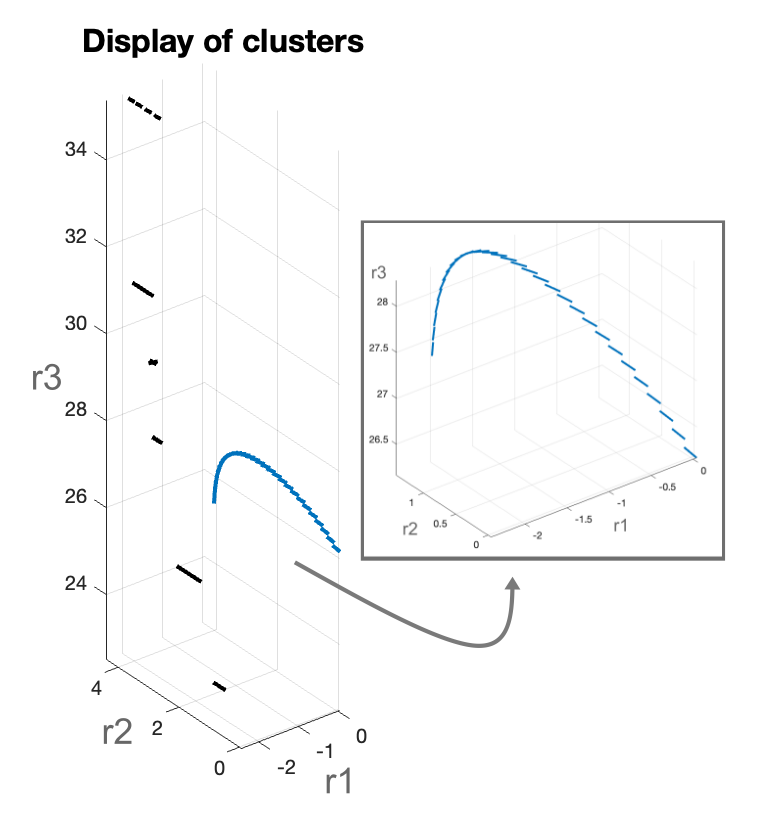}
         \caption{}
     \end{subfigure}
      \caption{ Results of the \texttt{SCSC} algorithm on the three-dimensional curve shown in Figure \ref{fig:stim-simple-curve}. (a) Display of the affinity matrix $\mathbf{P}$, where parameters for $\mathbf{J}$ are set to $T = 95$, $\varlambda = 0.0275$, $M = 400$, and $N = 10^5$. (b) Display of the spectrum of $\mathbf{P}$. The left image shows the spectrum $\{\lambda_i\}_{i = 1}^k$, while the right image illustrates $\{\lambda_i^\tau\}_{i = 1}^k$ with $\tau = 100$. The points selected in red on the right image correspond to the eigenvalues $\lambda_i^\tau$ that are greater than $1 - \varepsilon$, with $\varepsilon = 0.01$. The corresponding unexponentiated eigenvalues are shown in red on the left. Eigenvectors associated with these eigenvalues are used for the pre-clusters. (c) Display of pre-clusters, with each color representing a different cluster. (d) Display of the final clusters, with pre-clusters thresholded to include only those with more than $Q = 25$ elements. Blue points correspond to the only main cluster capturing the curve, while black points represent noise elements assigned to $C_0$.  }

        \label{1curvaGroup}
\end{figure}

\paragraph{Helix and arc.}

Next, we test the algorithm on a synthetic stimulus similar to the one described in \cite{AZ00}, and the results are illustrated in Figure \ref{fig:helixArc}.

Specifically, the stimulus, shown in image (a), consists of two synthetic images illustrating the 2D perspective projections of an arc (30 points) and an $r_2$-helix (60 points), which is a helix with its spiral extending along the $r_2$-axis. Image (b) shows the lifting in $\PO$, displaying all possible corresponding points.
Image (c) presents the affinity matrix $\mathbf{P}$, calculated using the chosen kernel parameters. Image (d) displays the spectrum of $\mathbf{P}$, with eigenvalues greater than $1 - \varepsilon$ (where $\varepsilon = 0.01$) highlighted in red.

\begin{figure}[H]
     \centering
     \centering
  \begin{subfigure}[b]{0.25\textwidth}
         \centering
         \includegraphics[width=\textwidth]{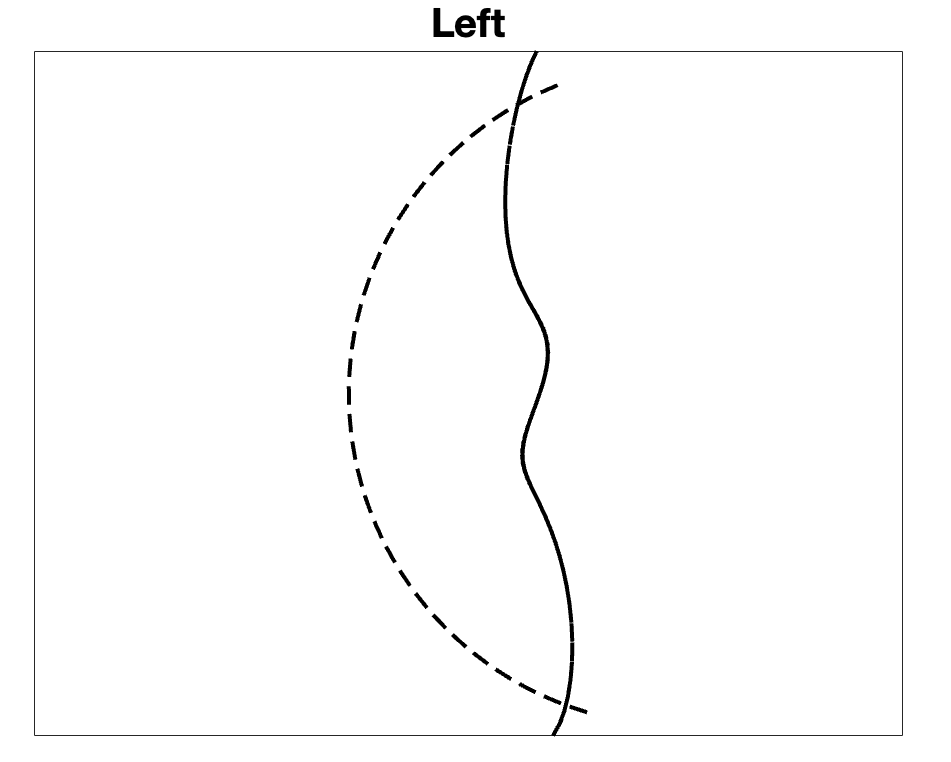}
         
         \includegraphics[width=\textwidth]{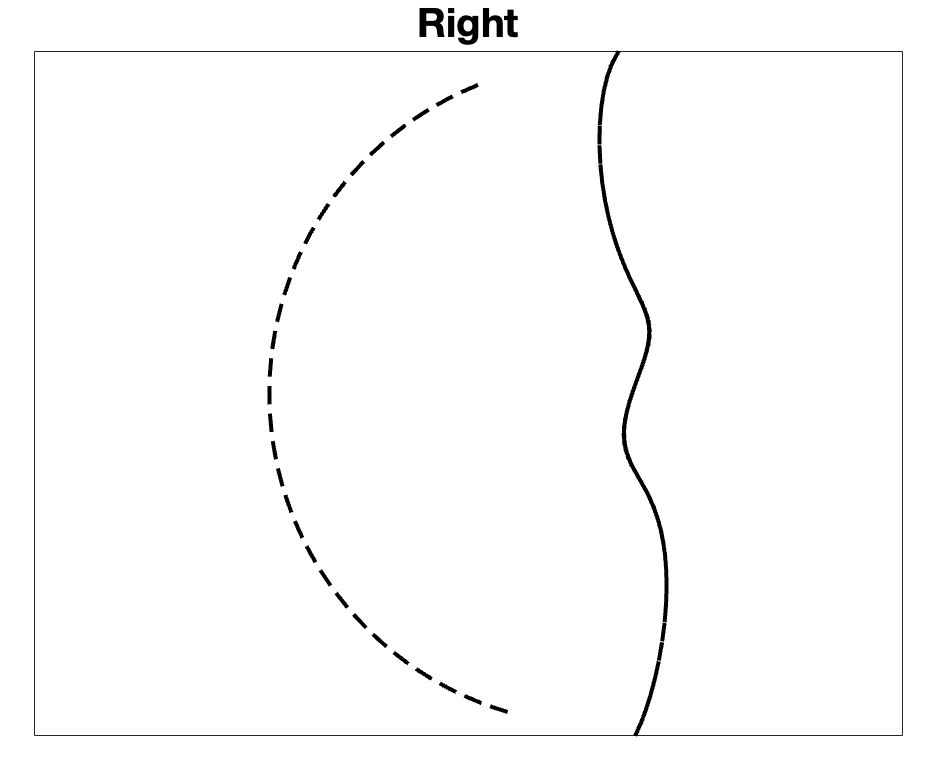}
         \caption{}
     \end{subfigure}
     \begin{subfigure}[b]{0.4\textwidth}
         \centering
         \includegraphics[width=\textwidth]{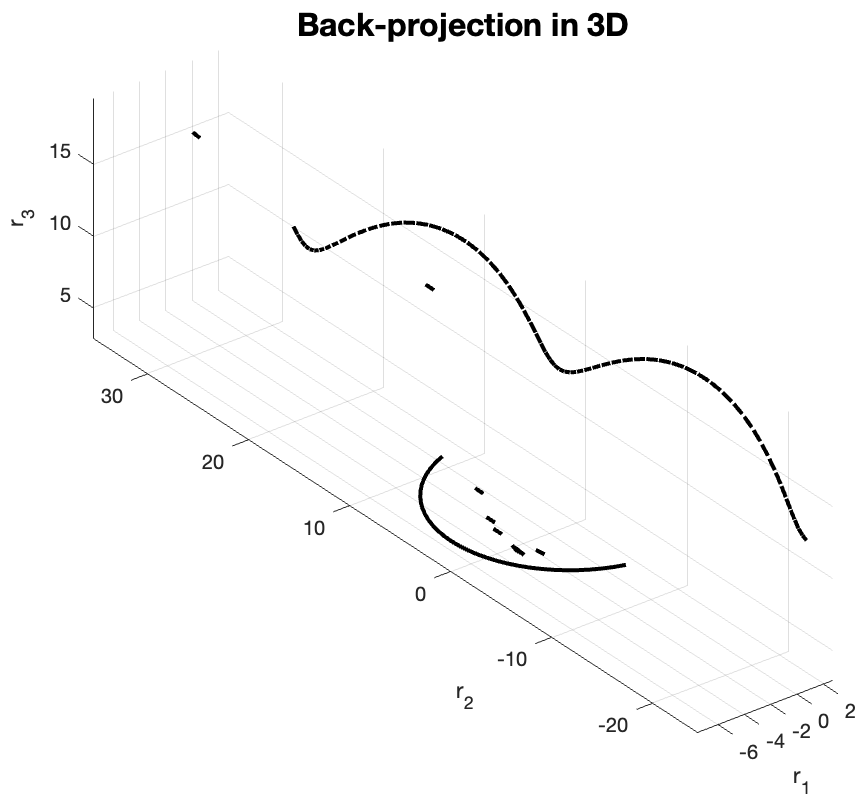}
         \caption{}
     \end{subfigure}
     
      \begin{subfigure}[b]{0.3\textwidth}
         \centering
         \includegraphics[width=\textwidth]{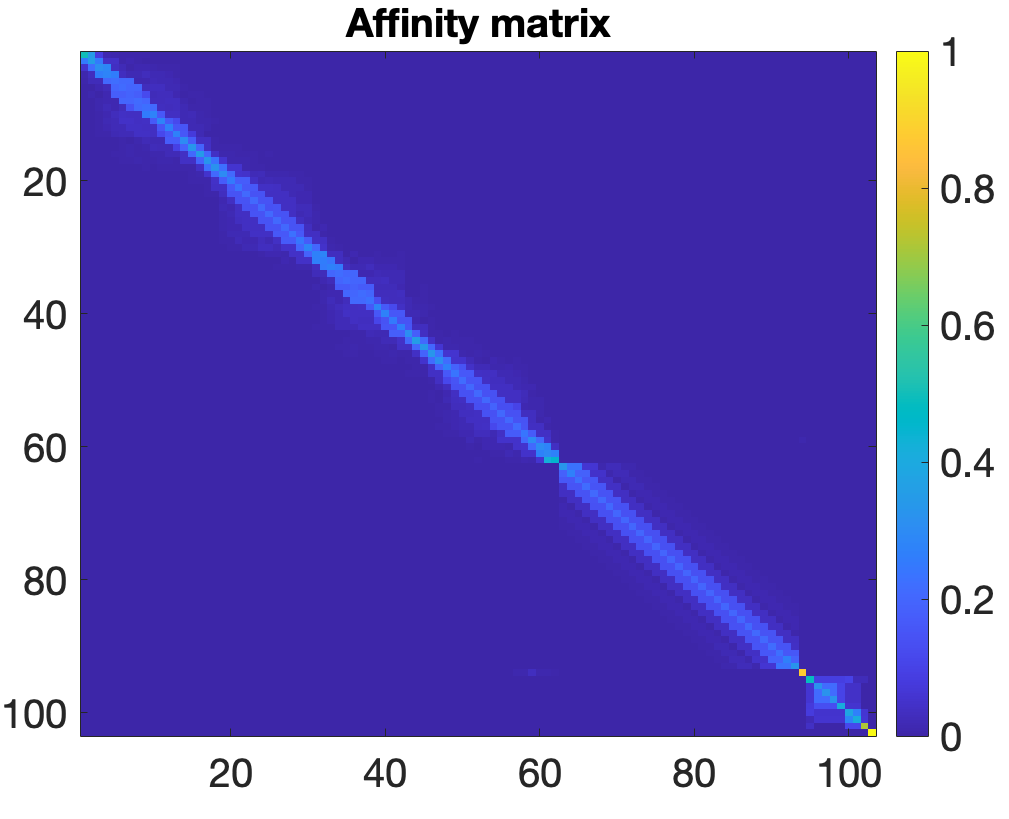}
         \caption{}
     \end{subfigure}
     \hspace{1cm}
     \begin{subfigure}[b]{0.3\textwidth}
         \centering
         \includegraphics[width=\textwidth]{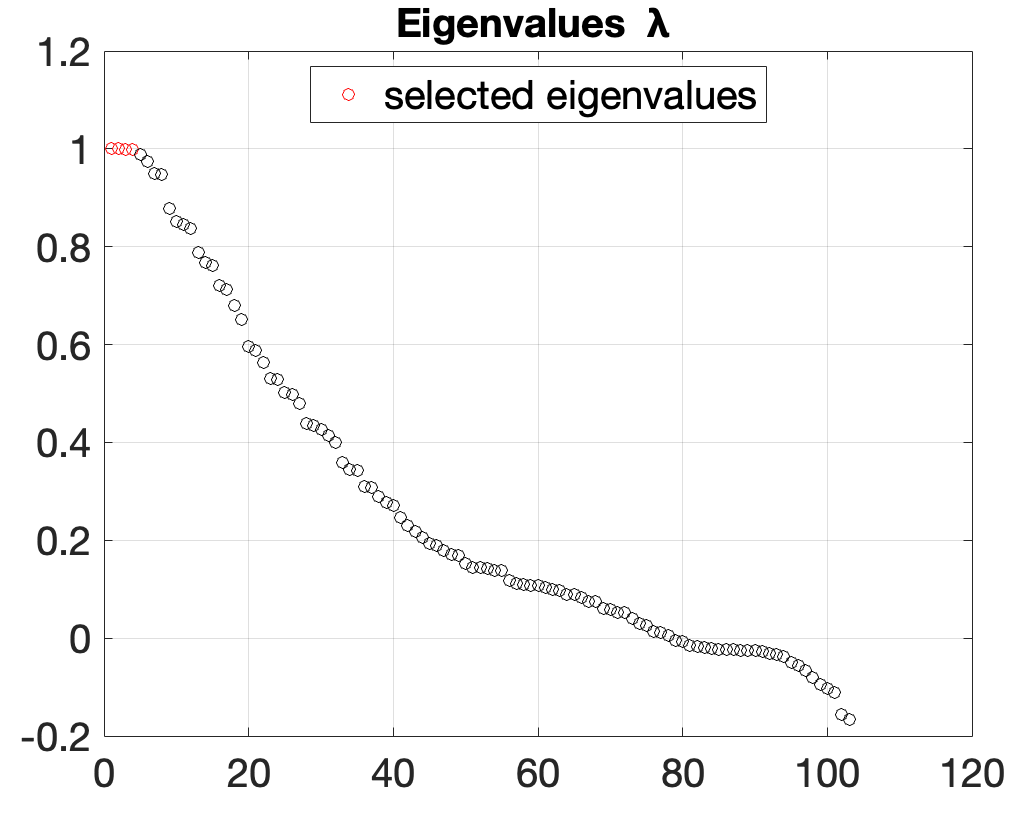}
         \caption{}
     \end{subfigure}

     \begin{subfigure}[b]{0.4\textwidth}
         \centering
         \includegraphics[width=\textwidth]{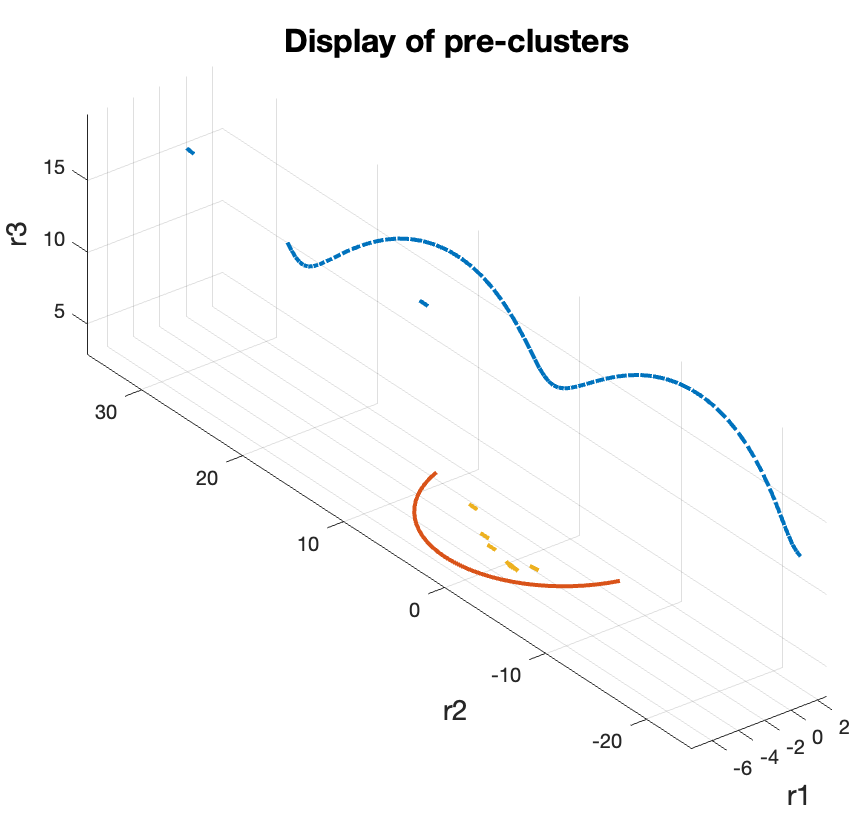}
         \caption{}
     \end{subfigure}
      \begin{subfigure}[b]{0.4\textwidth}
         \centering
         \includegraphics[width=\textwidth]{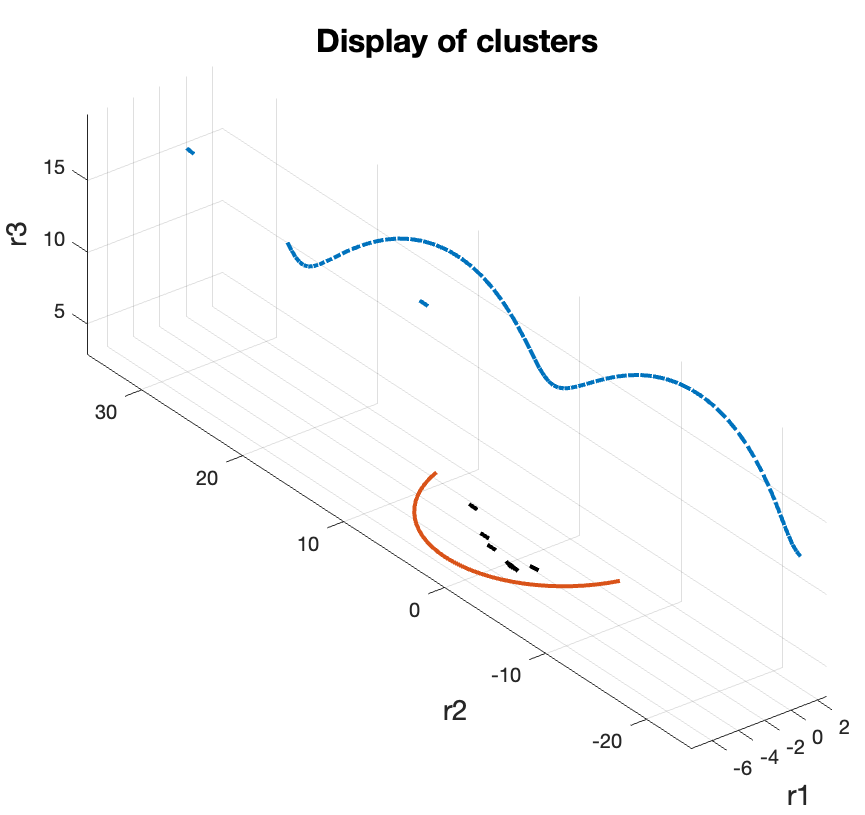}
         \caption{}
     \end{subfigure}
      \caption{Results of the \texttt{SCSC} algorithm on the three-dimensional arc and helix. (a) Left and right retinal images. (b) Lifting of the stimulus in $\PO$: the coupling generates the original 3D curves and noise. (c) Display of the affinity matrix $\mathbf{P}$, where parameters for $\mathbf{J}$ are set to $T = 100$, $\varlambda = 0.13$, $M = 400$, and $N = 10^5$. (d) Display of the spectrum of $\mathbf{P}$, colored in red are the selected eigenvalues  $\lambda_i> 1 - \varepsilon$, with $\varepsilon = 0.01$. The associated eigenvectors are used for the pre-clusters. (e) Display of pre-clusters, with each color representing a different cluster. (f) Display of final clusters, with pre-clusters thresholded to include only those with more than $Q = 20$ elements.   Blue points correspond to the cluster capturing the helix, while orange points belong to the cluster identifying the arc. Black points represent noise elements assigned to $C_0$.}
\label{fig:helixArc}
\end{figure}
Image (e) illustrates the pre-clusters, each represented by a different color. The points are organized into three main pre-clusters: a blue cluster corresponding to the points on the 3D helix, an orange cluster for the points on the 3D arc, and a yellow cluster containing points arising from false matches on the retinal planes. Finally, image (f) shows the clusters with a number of elements greater than $Q = 20$. We observe that the algorithm successfully segments the stimulus into two main clusters: one for the helix (blue points) and another for the arc (orange points). Noise points, displayed in black, belong to the $C_0$ cluster. Despite minor inaccuracies, such as two misclassified points from the helix (ladder points), the eigenvectors effectively capture the two primary perceptual units in the visual scene. Thus, the algorithm not only reconstructs the stereo images but also segments the 3D scene into distinct perceptual objects.

\subsubsection{Natural images}

Finally, we evaluate the algorithm on a pair of black-and-white natural images depicting twigs in a visual scene. These images are captured with a baseline of $6$ cm to simulate the distance between the left and right eyes. These images are shown on image (a) of Figure \ref{fig:gabor-filtering-twigs}. After applying a Gabor filtering process, which mimics the behavior of simple monocular cells that are selective for orientation in the primary visual cortex, we extract bidimensional positions and orientations. These points are displayed in Figure \ref{fig:gabor-filtering-twigs}, image (b). The points are then lifted into the space of 3D positions and orientations $\PO$. 

\begin{figure}[tbh]
     \centering
  \begin{subfigure}[b]{0.6\textwidth}
         \centering
         \includegraphics[width=\textwidth]{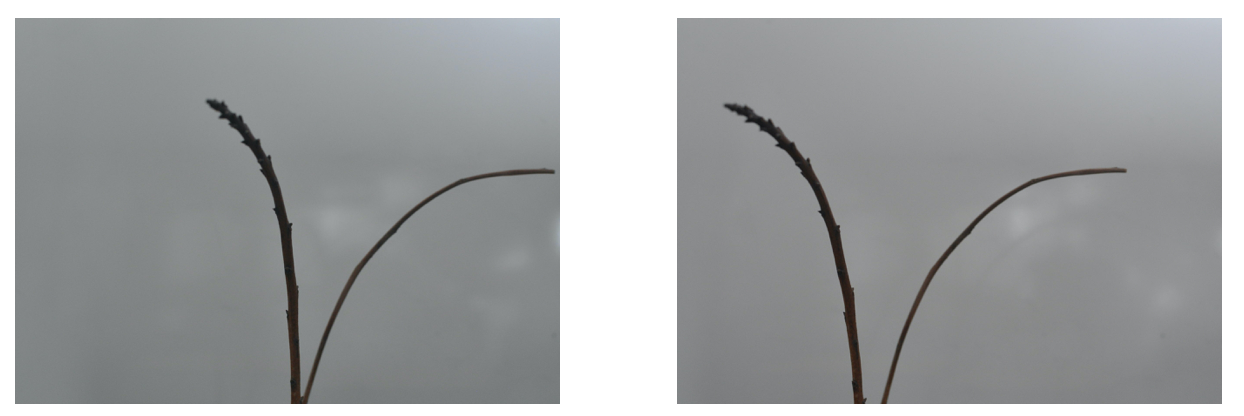}
         \caption{}
     \end{subfigure}
                
     \begin{subfigure}[b]{0.65\textwidth}
         \centering
         \includegraphics[width=0.45\textwidth]{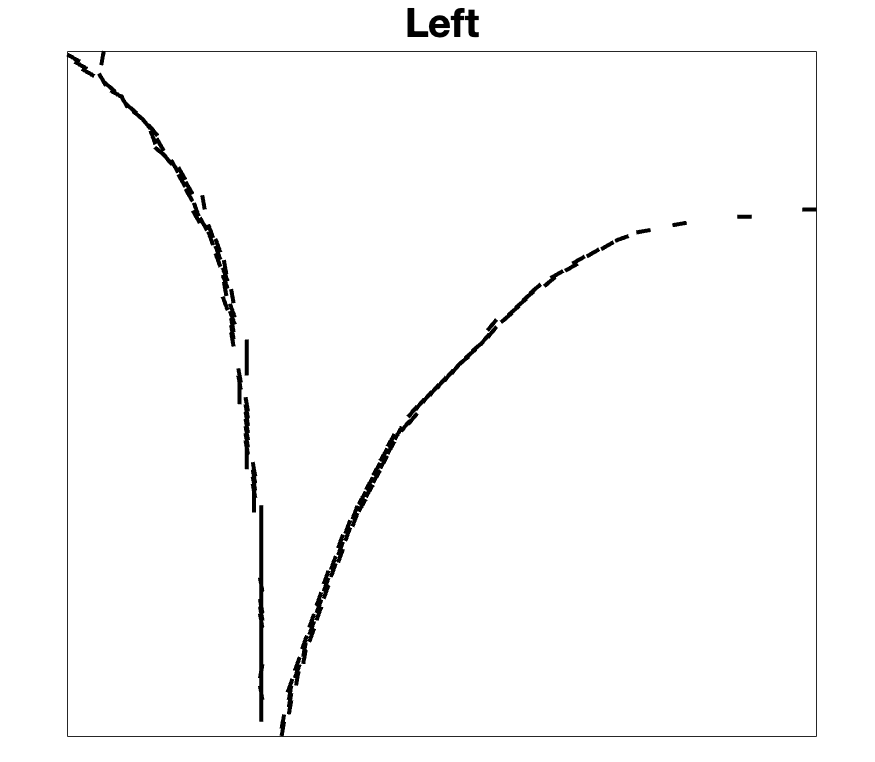}
         \includegraphics[width=0.45\textwidth]{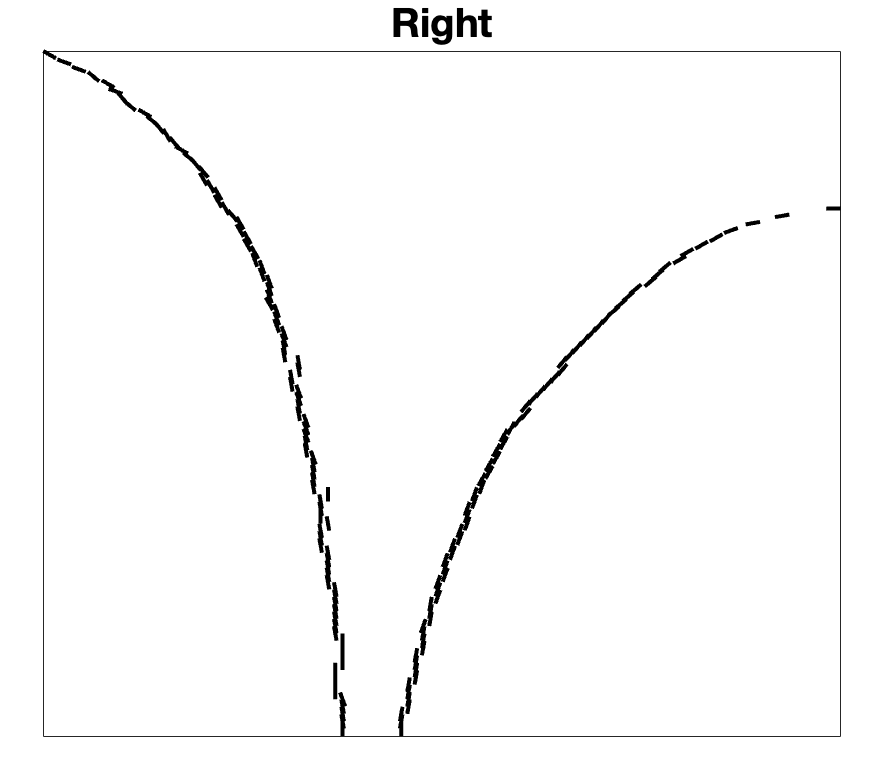}
         \caption{}
     \end{subfigure}
           \caption{Gabor filtering of twigs images. (a) Left and right natural images depicting twigs in a visual scene. (b) Result of Gabor filtering applied to these images, which lifts the points into the space $\mathbb{R}^2 \times \mathbb{S}^1$ of bidimensional positions and orientations. }
\label{fig:gabor-filtering-twigs}
\end{figure}

Results of the algorithm are illustrated in Figure \ref{fig:twigs-results}. Image (a) displays the lifting of the pair of retinal images into the space $\PO$, generating 3D representations of the twigs along with noise from false matches. Image (b) shows the affinity matrix $\mathbf{P}$, computed with kernel parameters for 
$\mathbf{J}$ :  $\varlambda = 0.0275$, $T = 40$, $M = 400$, and $N = 10^5$.
Image (c) presents the spectrum of $\mathbf{P}$, with eigenvalues greater than $1 - \varepsilon$ (where $\varepsilon = 0.01$) highlighted in red. Image (d) illustrates the pre-clusters, each depicted in a different color. Due to the high number of false matches, there is a large number of pre-clusters. However, two prominent groupings are evident, corresponding to the main twigs.
Indeed, image (e) shows the final clusters after applying a threshold of $Q = 50$. The algorithm successfully segments the stimulus into two main clusters, each corresponding to one of the twigs. Noise points, displayed in black, belong to the $C_0$ cluster.

\begin{figure}[H]
     \centering
     \begin{subfigure}[b]{0.3\textwidth}
         \centering
         \includegraphics[width=\textwidth]{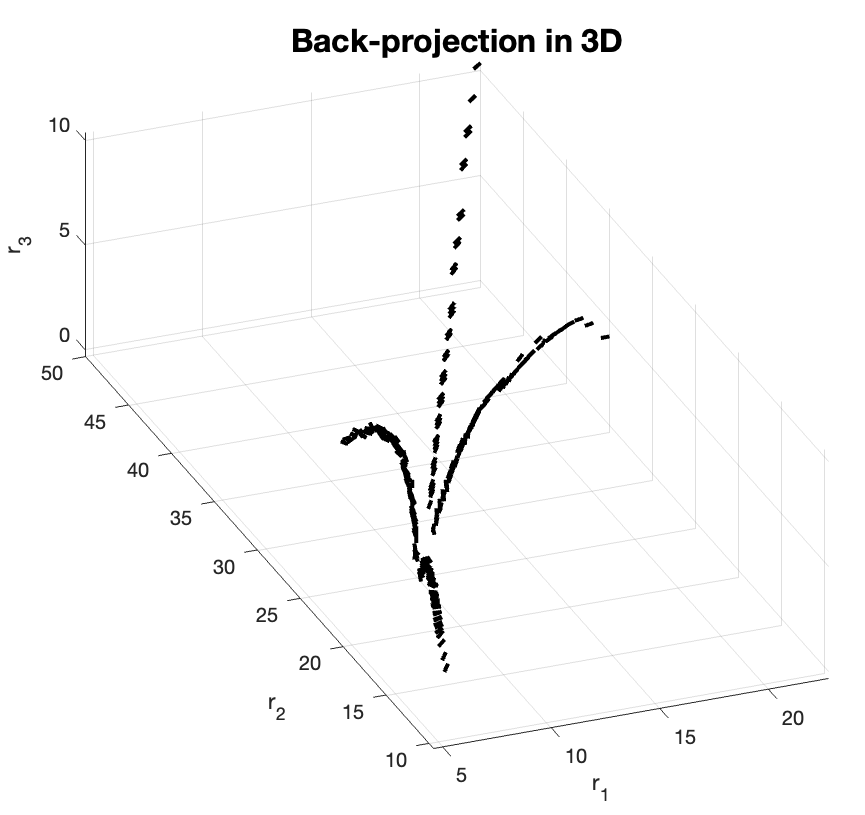}
         \caption{}
     \end{subfigure}
     \begin{subfigure}[b]{0.3\textwidth}
         \centering
         \includegraphics[width=\textwidth]{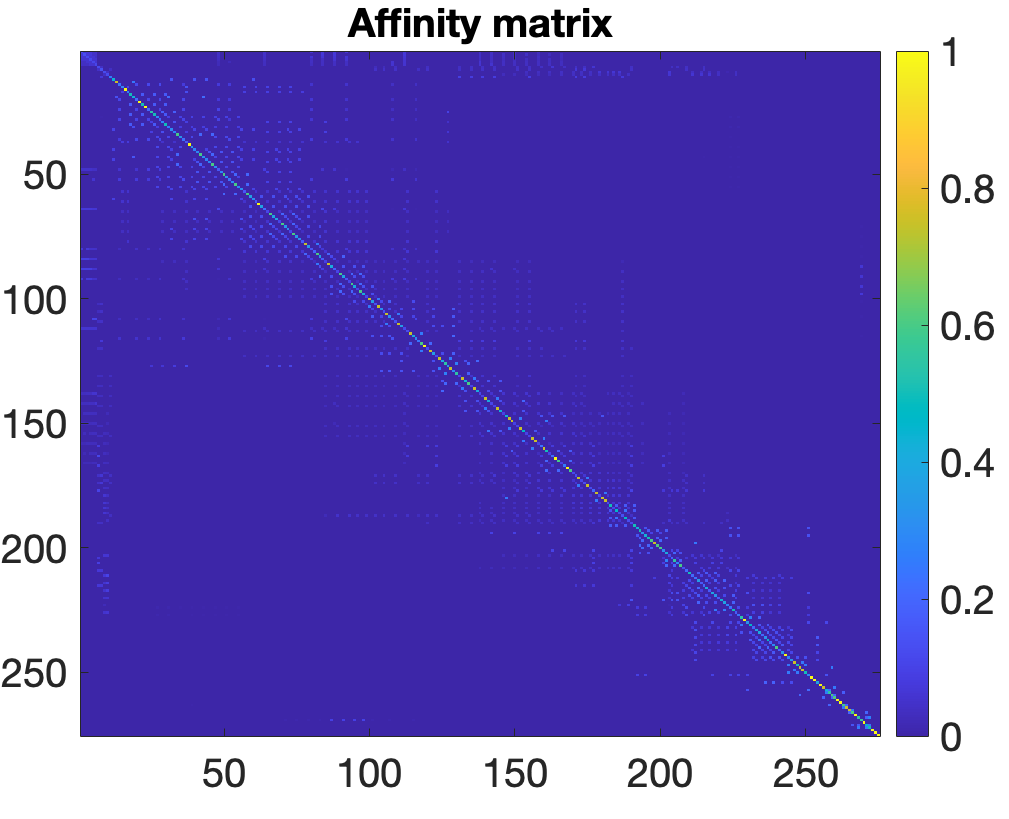}
         \caption{}
     \end{subfigure}
     \begin{subfigure}[b]{0.3\textwidth}
         \centering
         \includegraphics[width=\textwidth]{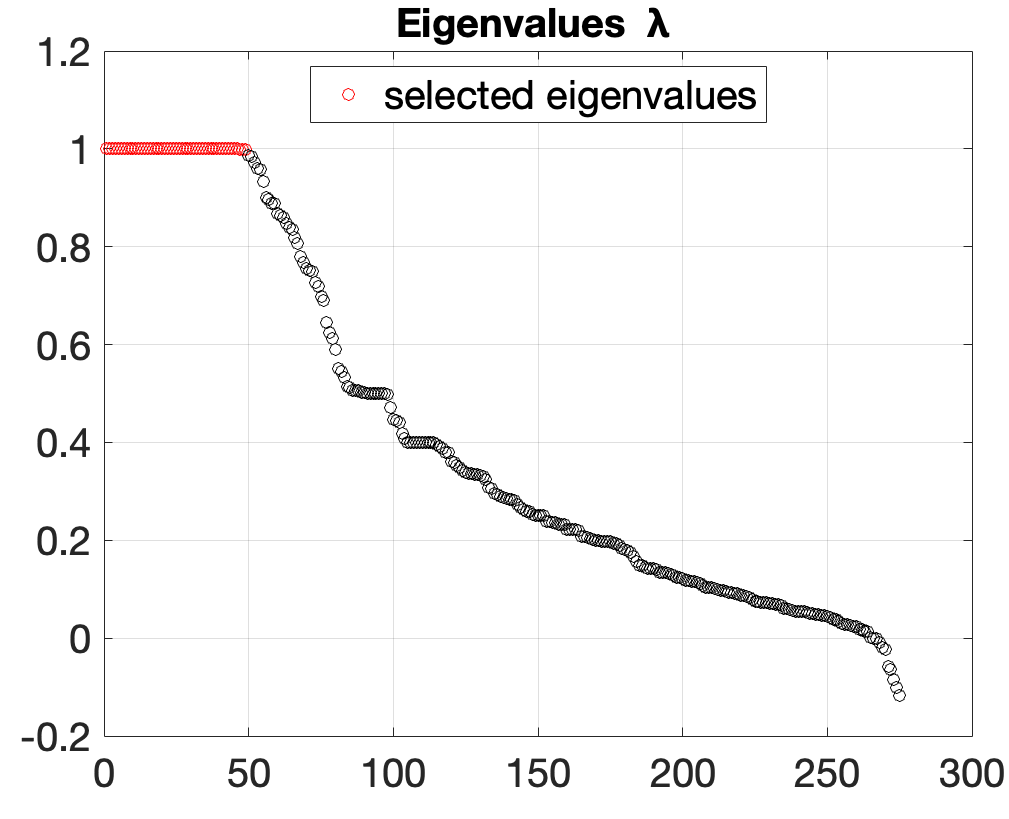}
         \caption{}
     \end{subfigure}

\vspace{0.5cm}

     \begin{subfigure}[b]{0.45\textwidth}
         \centering
         \includegraphics[width=\textwidth]{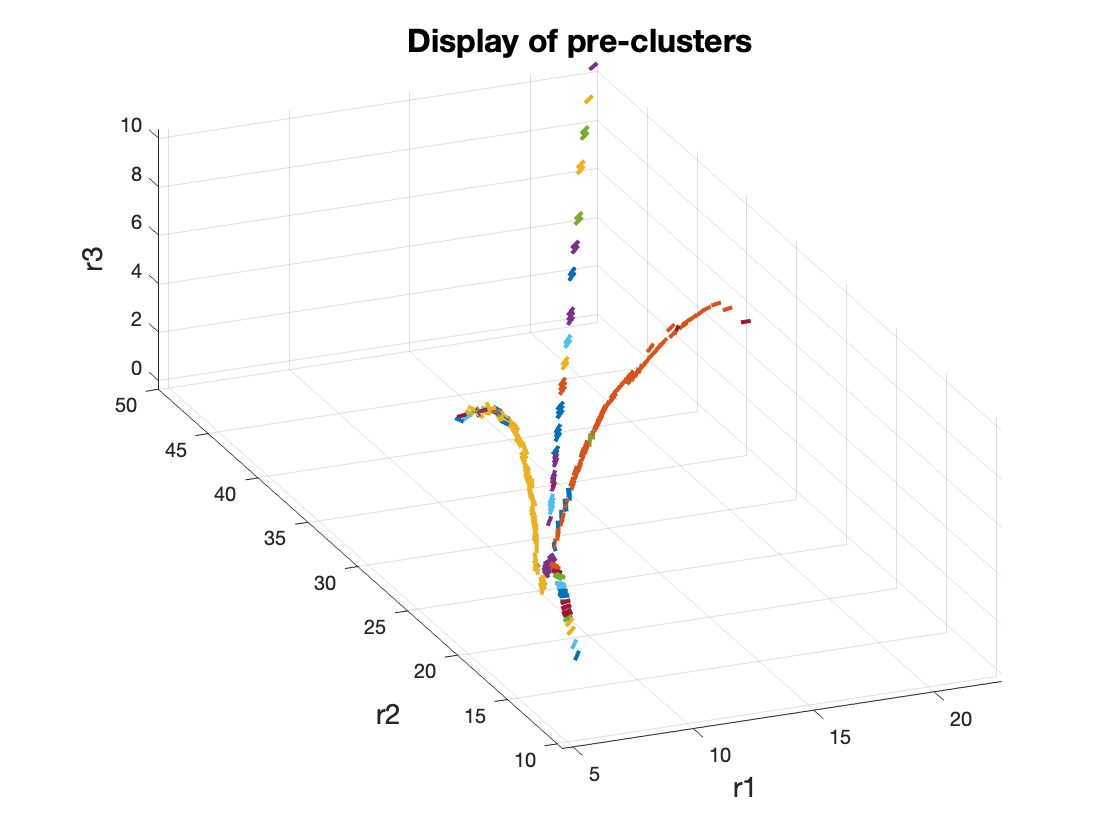}
         \caption{}
     \end{subfigure}
       \begin{subfigure}[b]{0.45\textwidth}
         \centering
         \includegraphics[width=\textwidth]{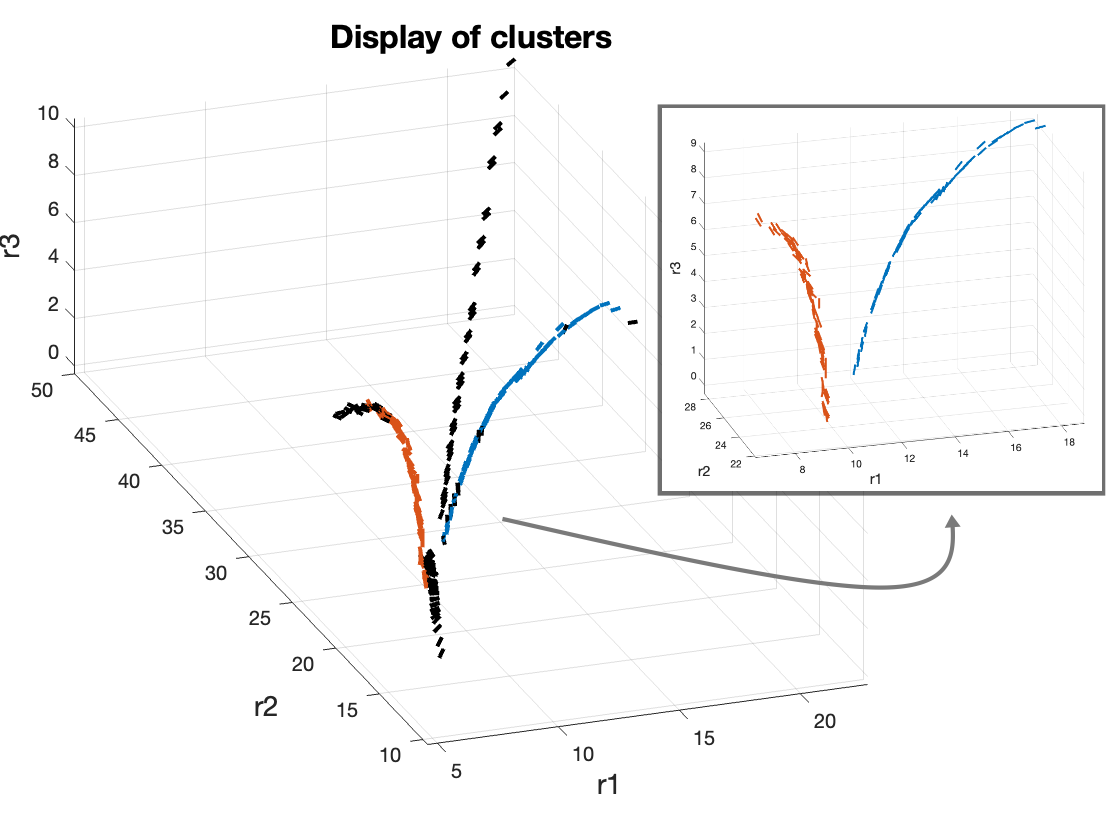}
         \caption{}
     \end{subfigure}
     
      \caption{Results of the \texttt{SCSC} algorithm on the natural images displaying twigs. (a) Lifting of the stimulus into $\PO$, showing both the original 3D twigs and noise generated from the image coupling.  (b) Display of the affinity matrix $\mathbf{P}$, where parameters for $\mathbf{J}$ are set to $T = 40$, $\varlambda = 0.0275$, $M = 400$, and $N = 10^5$. (c) Spectrum of $\mathbf{P}$, where eigenvalues greater than $1 - \varepsilon$ (with $\varepsilon = 0.01$) are highlighted in red. These eigenvalues are used to determine the pre-clusters. (d) Display of pre-clusters, each represented by a different color. (e) Final clusters after thresholding pre-clusters to include only those with more than $Q = 50$ elements. The only two main clusters, corresponding to the twigs, are highlighted in blue and orange, respectively. Black points indicate noise elements assigned to the $C_0$ cluster.}
\label{fig:twigs-results}
\end{figure}

\subsection{Comparison with Riemannian distance}

In this section, we compare the sub-Riemannian distance-based kernel we have introduced with the classical exponential kernel, which depends on the Riemannian distance of the space. The goal is to assess the performance and suitability of these kernels in different scenarios. We discuss how each kernel handles various types of data and how effectively they capture the underlying geometry of the space.

\subsubsection{Gaussian kernel}
In graph-based clustering (\cite{SM00, CG97, coifman2005geometric, CL2006}), it is common to use a Gaussian kernel as a similarity measure to group clouds of points. This kernel is typically a non-linear function of a Euclidean-type distance, defined as:
\begin{equation}
\label{GaussKernel}
k_{\xi_0, E}(\xi) = \frac{1}{4\pi\sigma}\exp\left( - \frac{d_{E}(\xi-\xi_0)
^2}{4\sigma}\right),
\end{equation}
where $d_{E}$ Euclidean-type distance on $\PO$, $\xi_0, \xi \in \PO$, with $\xi_0$ fixed, and $\sigma \in \R_+$. 

In this context, the Euclidean-type distance $d_{E}$ refers to a combination of the classical Euclidean distance $d_{\R^3}$ on $ \R^ 3$  and the classical Riemannian distance $d_{\S^2}$ on the sphere $\S^2$ (induced by the Euclidean metric in $\R^3$). Specifically:
\begin{equation}
d_{E} (\xi, \xi_0) = d_{\R^3}(p, p_0) + d_{\S^2}(n, n_0), \text{ with } \xi = (p, n) \in \PO.
\end{equation}
Isosurfaces in $\R^3$ of the Gaussian kernel \eqref{GaussKernel} are depicted in Figure \ref{cloud}. These are obtained after integration on $\S^2$ as done in \eqref{eq:intensityR3}.

This connectivity kernel can be viewed as the fundamental solution of the heat operator. The parameters of this kernel are the position $\xi \in \PO$ and the variance $\sigma \in \R_+$. Specifically, this kernel associates points that are close to each other in the Euclidean sense and decorrelates points that are farther apart. The parameter $\sigma$ acts as a scale parameter, influencing the extent of correlation: when $\sigma$ is small, the kernel provides higher correlation between points near $\xi_0$; conversely, when $\sigma$ is large, the kernel allows for correlation over a broader range, including points farther away.

\subsubsection{Numerical simulation}
 For starters, we test this measure to segment the clouds of points in image (b) of Figure \ref{cloud}. We apply the main steps of the algorithm outlined in Section \ref{sec:algorithm}, starting from these points in 3D. The only difference is that now the affinity matrix $\mathbf{J}$ is generated using the similarity measure defined by \eqref{GaussKernel}. The analysis reveals the presence of three main clusters, each representing a different cloud of points, as depicted in image (c): each color in the image corresponds to elements belonging to the same cluster.
 
 \begin{figure}[tbh]
     \centering
\begin{subfigure}[b]{0.3\textwidth}
         \centering
         \includegraphics[width=\textwidth]{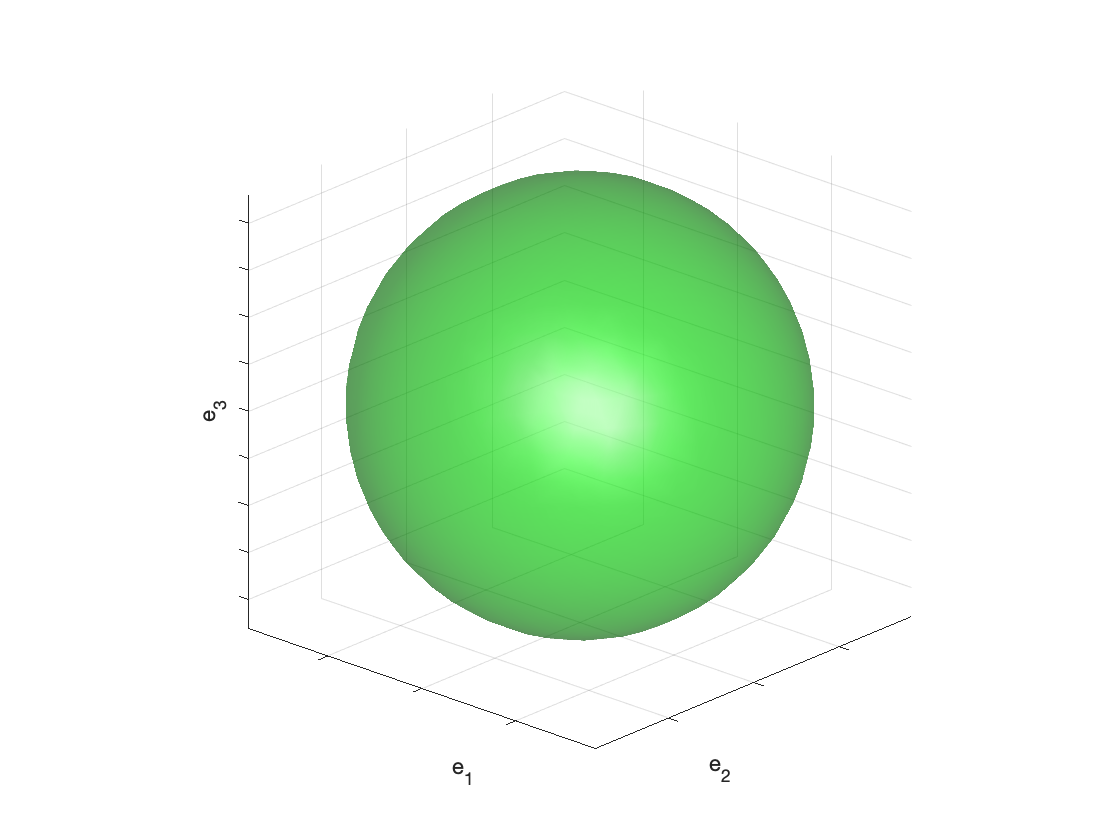}
      \caption{}
     \end{subfigure}
     \begin{subfigure}[b]{0.3\textwidth}
         \centering
         \includegraphics[width=\textwidth]{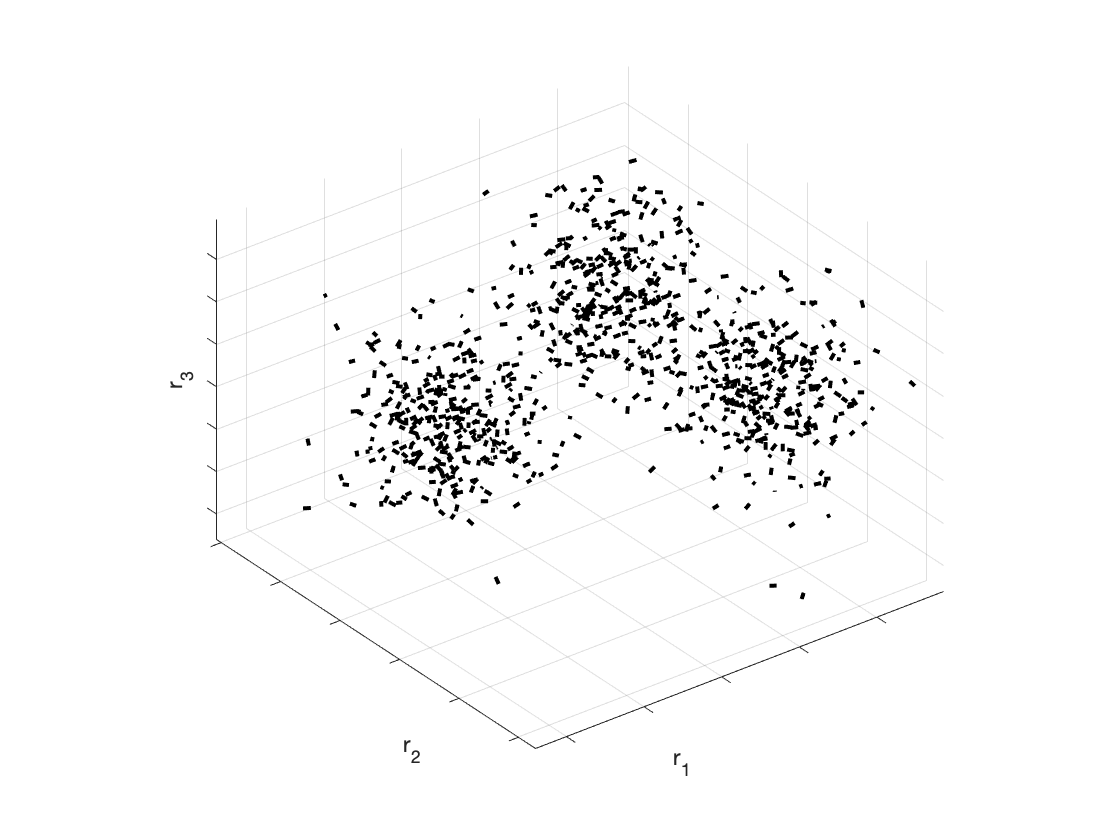}
         \caption{}
     \end{subfigure}
      \begin{subfigure}[b]{0.3\textwidth}
         \centering
         \includegraphics[width=\textwidth]{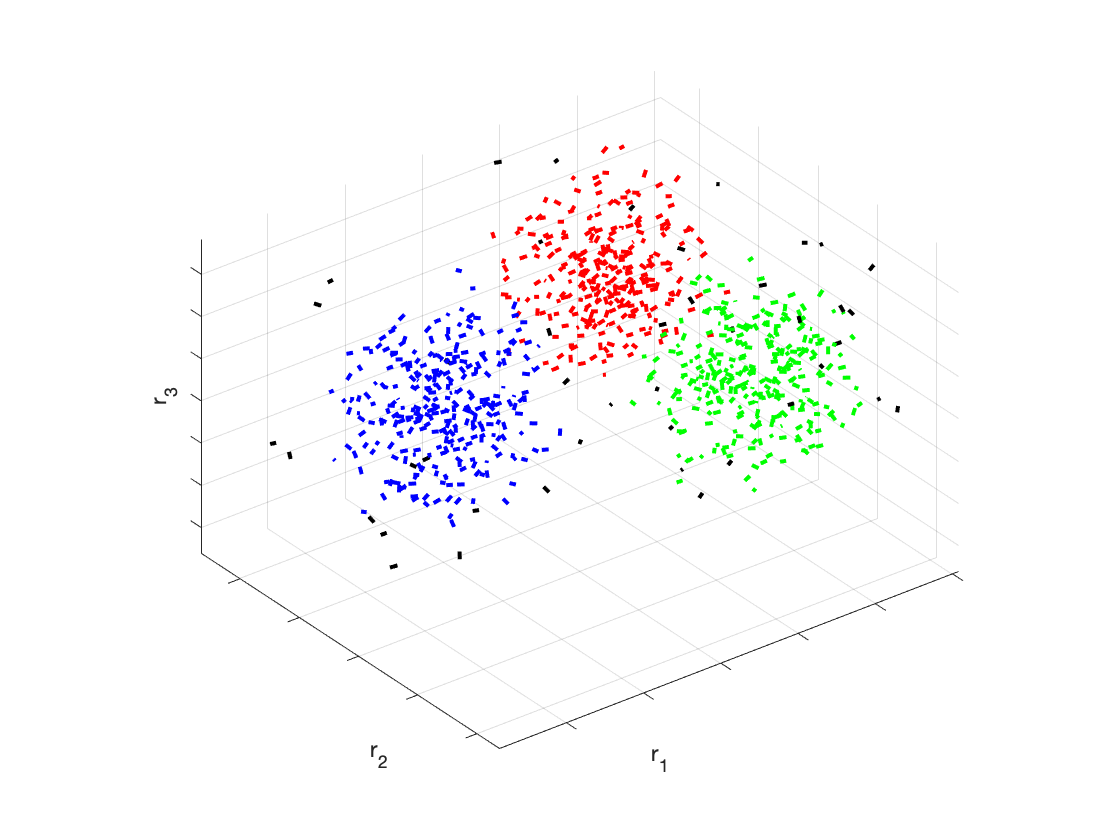}
         \caption{}
     \end{subfigure}
           \caption{Spectral analysis with ``Euclidean-type'' kernel.(a) Three-dimensional iso-surface of the Gaussian kernel defined by \eqref{GaussKernel}.  (b) Clouds of points in $\PO$. (c) Segmentation of the three clouds performed through the Gaussian kernel defined in \eqref{GaussKernel}. Elements with the same color belong to the same cluster.}
\label{cloud}
\end{figure}

However, this ``Euclidean'' kernel is not well-suited for solving stereo matching and the 3D reconstruction problem, because the points that constitute a true perceptual unit in 3D are not necessarily the ones that are closest or most densely clustered, but rather those that best follow the law of good continuation.

\begin{figure}[tbh]
     \centering
  \begin{subfigure}[b]{0.2\textwidth}
         \centering
         \includegraphics[width=\textwidth]{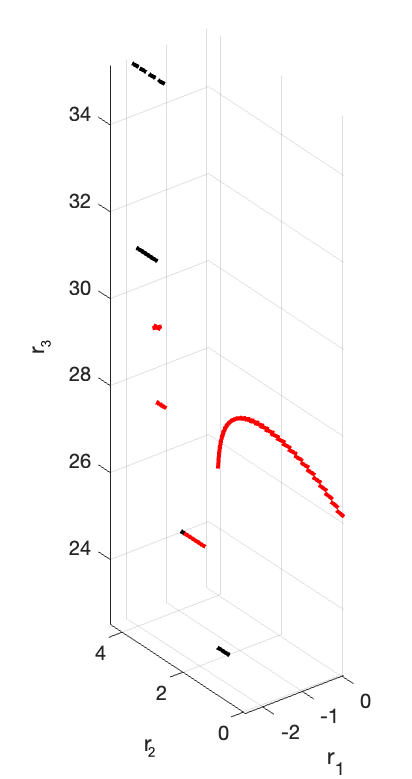}
         \caption{}
     \end{subfigure}
      \begin{subfigure}[b]{0.35\textwidth}
         \centering
         \includegraphics[width=\textwidth]{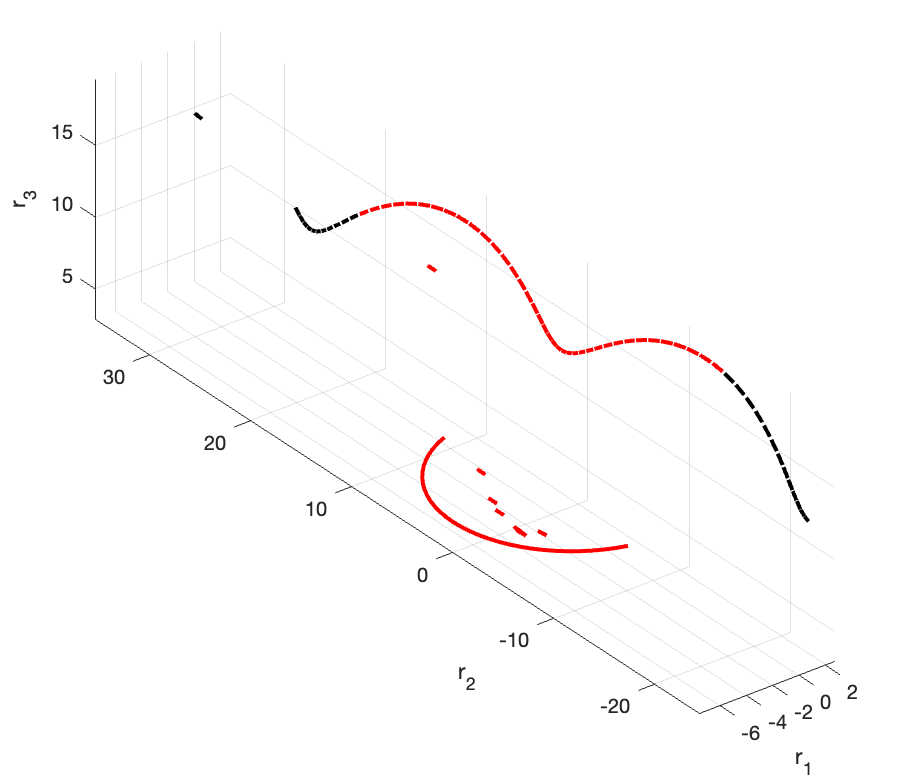}
         \caption{ }
     \end{subfigure}
     \begin{subfigure}[b]{0.3\textwidth}
         \centering
         \includegraphics[width=\textwidth]{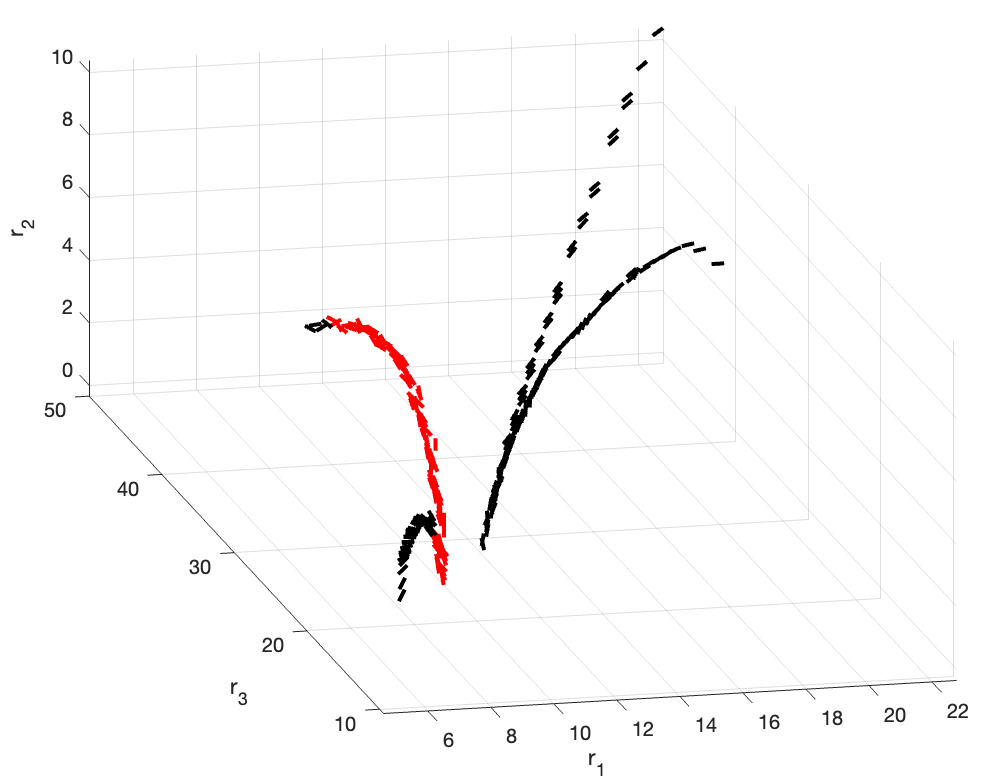}
         \caption{}
     \end{subfigure}
           \caption{Spectral analysis with ``Euclidean-type'' kernel.
           (a) Clustering performed on the 3D curve with $\sigma = 4$. (b) Clustering performed on the couple of arc and helix with $\sigma = 60$. (c) Clustering performed on twigs with $\sigma = 4$. }
\label{groupGauss}
\end{figure}

To illustrate, consider the previously defined three-dimensional synthetic images. When performing spectral clustering on the 3D parameterized curves, we observe that the primary cluster (red points) captures the region with the highest density of elements, but it fails to perform effective 3D reconstruction, as illustrated in Figure \ref{groupGauss}, image (a). In this case, adjusting the scale parameter $\sigma$ alone is insufficient for accurate 3D reconstruction. A similar result is observed with the helix and arc in 3D, depicted in image (b) of Figure \ref{groupGauss}, where the clustering only identifies the region with the highest number of points. When applying the same algorithm to the twig images, shown in Figure \ref{groupGauss}, image (c), the results are analogous. Here, the algorithm primarily groups points based on proximity, without effectively segmenting and reconstructing the visual scene. This demonstrates that the perceptual law underlying this grouping relies solely on proximity.

These examples illustrate that sub-Riemannian kernels are particularly effective in grouping three-dimensional perceptual units and addressing the stereo-matching problem, in contrast to Riemannian kernels, which are less effective in this context. Similar results were provided by Bekkers, Chen, and Portegies \cite{BCP18} investigating the use of sub-Riemannian and Riemannian distances in $SE(2)$ and $SE(3)$ for grouping of blood vessels.

Overall, sub-Riemannian distances provide superior performance compared to Euclidean-type distances, particularly in problems characterized by strong directionality. This underscores the advantage of sub-Riemannian methods for applications requiring preservation of directional information.

\section{Conclusions}

In this paper, we introduce a neurogeometric approach to stereo vision that aims to simultaneously identify 3D perceptual units and solve the stereo correspondence problem. This study builds upon and extends the 
neurogeometric model presented in \cite{BCSZ23} and \cite{BCSZ23b}, establishing a connection between the functional architecture of binocular cells in V1, the psychophysics of 3D association fields, and the emergence of visual percepts in space.

We show how the stochastic counterpart of the 3D association fields enables transitions from points in space to nearby points, which can be conceptualized as underlying circuits of binocular neurons. This leads to the introduction of a connectivity kernel, a measure that encodes the 3D geometric position-orientation information and describes the probability of co-occurrence between elements in space. These probability densities are incorporated as facilitation inducers in a neural population activity model, whose stability analysis results in the emergence of three-dimensional perceptual units.

We employ a clustering algorithm to evaluate the visual grouping properties of these discrete connectivity kernels. By associating all possible corresponding points (left and right retinal points with the same abscissa coordinate) and lifting them into points $\xi_i \in \Omega \subset \PO$, we also account for false matches, i.e., points not belonging to the original stimulus.

The proposed algorithm effectively correlates 3D real elements, grouping points that pertain to a single object in 3D space, and solves the stereo correspondence problem by accurately identifying correct matches in the  left and right images. During this process, false matches are eliminated as the similarity measure introduced by the kernel associates elements that satisfy the good-continuation constraints.

A comparison with the same algorithm using the classical Gaussian kernel with a Euclidean-type distance highlights and justifies the sub-Riemannian metric. This approach seamlessly integrates spectral clustering techniques with the Gestalt principle of good continuation, providing a robust solution that not only addresses the stereo reconstruction challenge but also achieves effective segmentation of the visual scene.

\section*{Acknowledgments}
SWZ was supported by NIH grant 1R01EY031059 and NSF Grant 1822598. GC is supported by the project PRIN 2022 F4F2LH - CUP J53D23003760006, and by MNESYS PE12 (PE0000006).

\printbibliography

\end{document}